\pdfoutput=1
\documentclass[acmsmall, screen, nonacm, natbib=false]{acmart}
\usepackage{preamble}
\usepackage{jon-tikz}
\usepackage{iblock}
\usepackage{macros}
\usepackage{categories}

\title{Directed univalence in simplicial homotopy type theory}

\author{Daniel Gratzer}
\orcid{0000-0003-1944-0789}
\email{gratzer@cs.au.dk}
\affiliation{
  \institution{Department of Computer Science, Aarhus University}
  \country{Denmark}
}

\author{Jonathan Weinberger}
\orcid{0000-0003-4701-3207}
\email{jweinberger@chapman.edu}
\affiliation{
  \institution{Fletcher Jones Faculty Fellow, Fowler School of Engineering, Schmid
College of Science and Technology, Center of Excellence in Computation, Algebra, and Topology (CECAT),
Chapman University}
  \country{USA}
}

\author{Ulrik Buchholtz}
\orcid{0000-0002-5944-6838}
\email{ulrik.buchholtz@nottingham.ac.uk}
\affiliation{
  \institution{School of Computer Science, University of Nottingham}
  \country{United Kingdom}
}

\date{\today}

\setcopyright{none}

\begin{document}

\begin{abstract}
  Riehl and Shulman's simplicial type theory extends homotopy type theory with a
  directed interval type,
  allowing it to be modeled in simplicial spaces (and simplicial objects in any higher topos).
  The main application is the development of synthetic higher category theory,
  modeling $(\infty,1)$-categories as types satisfying a complete Segal condition.

  We extend simplicial type theory with modalities and new
  reasoning principles to obtain \emph{triangulated type theory} in order to construct
  the $(\infty,1)$-category of spaces $\Space$,
  and from it many other concrete $(\infty,1)$-categories.

  We prove that homomorphisms in $\Space$ correspond to ordinary
  functions of types, \ie{}, that $\Space$ is \emph{directed} univalent.
  From this we can significantly extend the reach of synthetic higher category theory
  with more results and examples,
  including the first complete examples of
  the \emph{structure homomorphism principle},
  a directed version of the structure identity principle known from homotopy type theory.
\end{abstract}
\maketitle

\section{Introduction}
\label{sec:introduction}
Homotopy type theory (\HOTT{}) is a type theory for synthetic $\infty$-groupoid theory;
it can be modeled in, and hence serves as an internal language for,
any Grothendieck $(\infty,1)$-topos~\parencite{shulman:2019}.
It builds on Martin-Löf's dependent type theory by adding Voedvodsky's univalence axiom
and a range of higher inductive types~\parencite{hottbook}.
Martin-Löf's identity types equip every type
with a proof-relevant coherent equivalence relation which is
respected by every construction in type theory~\parencite{lumsdaine:2009}.
\HOTT{} has proven useful as a tool for synthetic homotopy theory
and is well suited for formalization using proof assistants.

Proof assistants are well-tuned to support replacing equal elements by equal elements, where
equality is reified by the intensional identity type within type theory. Accordingly, if two
distinct terms can be identified, they can be swapped out for each other in large proofs without
further effort. In \HOTT, the identity type becomes far richer and, in particular, elements of the
universe become identified whenever they are equivalent. Accordingly, users of proof assistants
based on \HOTT{} can swap out \eg{}, an implementation of the integers well-suited for reasoning
with an equivalent version tuned for efficient computation without additional effort.  This offers
the same convenience to types that function extensionality grants functions.
\textcite{angiuli:sip:2021}, for instance, show that this can be used to internalize some applications
of parametricity but, crucially, without eliminating standard models which do not support the full
apparatus of parametricity.

A type theory for groupoids makes it far easier to manipulate equality, but what about formalization
challenges which are fundamentally asymmetric? For a toy example, consider an algorithm traversing a
list to sum its elements $\Con{sum} : \prn{A : \Con{Monoid}} \to \Con{List}\,A \to A$. Univalence
and one of its important consequences, the structure identity principle, tell us that $\Con{sum}$
must respect monoid isomorphisms. But far more is true: $\Con{sum}$ commutes with all monoid
homomorphisms. To prove this we must \hypertarget{firstReq}{(1)} formulate how a monoid homomorphism
$f : A \to B$ induces a map $\Con{List}\,f : \Con{List}\,A \to \Con{List}\,B$ and
\hypertarget{secondReq}{(2)} show that $\Con{sum} \circ \Con{List}\,f = f \circ \Con{sum}$.  Neither
task follows from univalence as $f$ need not be invertible and univalence handles only symmetric
relations.

\subsection{A type theory for categories}

The above example would be possible in a version of type theory where types encoded not just
groupoids but \emph{categories}: a directed type theory. That is, each type would come equipped with
a notion of homomorphism (along with composition, \etc{}) and each term in the type theory would be
bound to automatically respect homomorphisms \eg{}, be functorial.  Aside from the benefits to
formalization, it is particularly desirable to find a directed version of \HOTT{} where types would
encode $\infty$-categories\footnote{Specifically, $(\infty,1)$-categories: categories whose morphisms form an
  $\infty$-groupoid.}~\parencite{joyal:2008,lurie:2009,cisinski:2019,riehl:2022};
$\infty$-category theory is an important area of mathematics but whose foundations are well-known to
be cumbersome. It is conjectured that directed homotopy theory could serve as the basis for a more
usable and formalizable foundation of this field. Many such theories (both homotopical and not) have
been studied over the years~\parencite{%
  licata:2011,%
  warren:2013,%
  nuyts:2015,%
  north:2018,%
  kavvos:directed:2019,%
  nuyts:2020,%
  weaver:2020,%
  ahrens:2023,%
  kolomatskaia:2023,%
  neumann:2024,
  neumann:2025}.

A key obstruction to this program is that ($\infty$-)categories do not behave well enough to support
a model of type theory where every type is a category. For instance, $\Pi$-types do not always exist
because the category of categories is not locally cartesian closed. Most directed type theories
therefore change how type theory works to \eg{}, allow only certain kinds of $\Pi$-types and
dependence. We will focus on a different approach introduced by \textcite{riehl:2017}: \emph{simplicial
  type theory (\STT{})}. The key insight is to not require that every type is an
$\infty$-category, but instead a more flexible object from which we can carve out genuine
$\infty$-categories using two definable predicates.

\STT{} extends \HOTT{} with a new type to probe the implicit categorical structure each type
possesses: the \emph{directed interval} $\Int$. \textcite{riehl:2017} further equip it with the
structure of a bounded linear order $\prn{\land,\lor,0,1}$. One can then use $\Int$ to access \eg{}, the
morphisms $a$ to $b$ in $A$ by studying ordinary functions within type theory $f : \Int \to A$ such
that $f\prn{0} = a$ and $f\prn{1} = b$.

Early
evidence~\parencite{riehl:2017,riehl:2023,riehl:2025,buchholtz:2019,buchholtz:2023,weinberger:twosided:2024,weinberger:sums:2024,weinberger:chevalley:2024,bardomiano:2024,bardomiano:2025}
suggests that simplicial type theory approaches the desired usable foundations for $\infty$-category
theory. A number of definitions and theorems from classical $\infty$-category theory have been
ported to \STT{} and the proofs are shorter and more conceptual. Even better, Kudasov's
experimental proof assistant \RZK{}~\cite{kudasov:23} for \STT{} has shown that the arguments for
\eg{}, the Yoneda lemma are simple enough to be formalized and machine-checked~\parencite{kudasov:2024}.

\begin{convention}
  For the remainder of this paper, we shall be concerned only with $\infty$-categories and
  constructions upon them. Accordingly, hereafter we largely drop the ``$\infty$-'' prefix and speak
  simply of categories, groupoids, \etc{} except in those few situations where it would cause
  ambiguity.
\end{convention}

\paragraph{Simplicial type theory, a reprise}
A brief description of simplicial type theory is in order.  As mentioned, every type $A$ in \STT{}
has a notion of homomorphism: functions $\Int \to A$. However, in arbitrary types these do not
really behave like homomorphisms \eg{}, they need not compose.

Suppose we are given $f,g : \Int \to A$ such that $f\,1 = g\,0$. A composite $h$ ought to be a
homomorphism such that $h\,0 = f\,0$ and $h\,1 = g\,1$, but not every such $h$ satisfying just these
conditions ought to be a composite. In particular, further data is required to connect $h$ with $f$
and $g$. Classically, all of this is encoded by a 2-simplex $H$ (see the left diagram in
\cref{fig:intro:simplices}). Inside simplicial type theory, we represent such 2-simplices as maps
$\Delta^2 \to A$ where $\Delta^2 = \Compr{\prn{i,j} : \Int \times \Int}{i \ge j}$ (c.f., the shaded
portion of the right-hand diagram of \cref{fig:intro:simplices}).

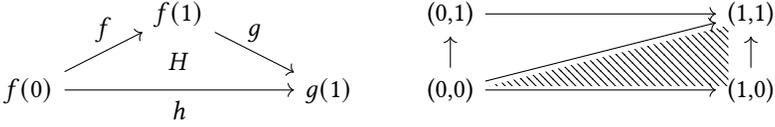
\begin{figure}
  \[
    \begin{tikzpicture}[diagram]
      \node (A) {$f\prn{0}$};
      \node[above right = 1cm and 2cm of A] (B) {$f\prn{1}$};
      \node[right = 4cm of A] (C) {$g\prn{1}$};
      \node[above right = 0.35cm and 2cm of A] (D) {$H$};
      \path[->] (A) edge node[above] {$f$} (B);
      \path[->] (B) edge node[above] {$g$} (C);
      \path[->] (A) edge node[below] {$h$} (C);
    \end{tikzpicture}
    \qquad
    \begin{tikzpicture}[diagram]
      \fill[pattern=north west lines] (0.5,0.05) -- (3.7,0.05)  -- (3.7,0.85) -- (0.5,0.05);
      \node[] (A) at (0,0) {(0,0)};
      \node [right = 4cm of A] (B) {(1,0)};
      \node [above = 1cm of A] (C) {(0,1)};
      \node [above right = 1cm and 4cm of A] (D) {(1,1)};
      \path[->] (A) edge (B);
      \path[->] (A) edge (C);
      \path[->] (C) edge (D);
      \path[->] (B) edge (D);
      \path[->] (A) edge (D);
    \end{tikzpicture}
  \]
  \caption{Illustrations of simplices}
  \label{fig:intro:simplices}
\end{figure}

In particular, a 2-simplex $H : \Delta^2 \to A$ witnesses that $H\prn{-,0}$ and $H\prn{1,-}$ can be
composed to obtain $\lambda i.\,H\prn{i,i}$. It is convenient to isolate the subtype
$\Lambda^2_1 = \Compr{\prn{i,j}}{i = 1 \lor j = 0} \subseteq \Delta^2 \subseteq \Int \times
\Int$. Unfolding, a map $\Lambda^2_1 \to A$ corresponds to a pair of composable arrows $f,g$.
Accordingly, every pair of composable arrows in $A$ admits a unique composite \ie{}, $A$ is
\emph{Segal} if $\prn{\Delta^2 \to A} \Equiv \prn{\Lambda^2_1 \to A}$.

Segal types already possess enough structure to behave like categories \eg{} it follows that
composition is associative and unital just from the Segal condition. Unfortunately, they may suffer
from an excess of data: they come equipped with two notions of sameness. Namely, $a,a' : A$ may be
regarded as the same when $a = a'$ \emph{or} when there is an invertible homomorphism from $a$ to
$a'$. In keeping with our pursuit of the structure homomorphism principle, we shall be interested in
types where these coincide \ie{} where $A \to \Sum{f : A^\Int}\IsIso\prn{f}$ is an equivalence. We
say such a type is \emph{Rezk} if it is Segal and satisfies this condition. An important result of
\textcite{riehl:2017} is that Rezk types adequately model the standard notion of
$\infty$-category~\parencite{rezk:2001}: a theorem proven in \STT{} about a Rezk type describes a valid result for ordinary $\infty$-categories.\footnote{In fact, combined with general results on
  \HOTT{}~\parencite{shulman:2019,weinberger:2022} they model \emph{internal} $\infty$-categories in an
  $\infty$-topos~\parencite{martini:2022,martini:2023,cisinski:2024}.\footnote{For a discussion on \emph{non-standard} models see~\parencite{rasekh:2025}}}

\paragraph{Directed univalence}
However, simplicial type theory is not a panacea for replacing classical $\infty$-category
theory. Presently, it is really only suitable for studying ``formal'' questions and, surprisingly,
it is unknown how to construct a non-trivial closed Rezk type within \STT{}. Crucially,
\STT{} lacks an equivalent to the category of groupoids (the $\infty$-categorical version of the
category of sets). Not only does this mean that \STT{} faces severe limitations on what theorems can
be \emph{stated}, it is presently impossible to exploit directed path types when formalizing.
Returning to our original example with $\Con{sum}$, \STT{} would automatically handle
\hyperlink{firstReq}{(1)} and \hyperlink{secondReq}{(2)} if there was a type of monoids
$\Con{Monoid}$ where directed paths were monoid homomorphisms, but such a definition is presently
out of reach.

Our central contribution is to overcome these challenges by extending \STT{} with new reasoning
principles and constructing a Rezk type $\Space$ whose objects correspond to groupoids (\ie{}, Rezk
types where every homomorphism is invertible) and whose homomorphisms are functions. This last
requirement is termed \emph{directed univalence}:
\begin{definition}
  \label{def:intro:dua}
  A universe $\Space$ is \emph{directed univalent} if $\Int \to \Space$ is isomorphic to
  $\Sum{A\,B : \Space} A \to B$ over $\Space \times \Space$.
\end{definition}

Before discussing our approach, we survey a few consequences of this result. Once $\Space$ is
available, a number of applications of \STT{} snap into focus. For instance, one can isolate
subcategories of $\Space$ such as the category of sets $\Space_{\le 0}$ and the category of
propositions $\Space_{\le -1}$. Using the ordinary constructions of type theory, one can parlay
these into our aforementioned category of monoids:
\[
  \Con{Monoid} =
  \Sum{A : \Space_{\le 0}} \Sum{\epsilon : A} \Sum{\cdot : A \times A \to A}
  \Con{isAssociative}\prn{\cdot} \times \Con{isUnit}\prn{\cdot,\epsilon}
\]
The only difference in this definition from the standard one seen in ordinary type theory is the
replacement of $\Uni$ by $\Space_{\le 0}$. However, with just this change we are able to prove the
following result:
\begin{restatable*}{lemma}{monoidnat}
  If $F,G : \Con{Monoid} \to \Space$ and $\alpha : \prn{A : \Con{Monoid}} \to F\prn{A} \to G\prn{A}$
  then $\alpha$ is natural \ie{} if $f : A \to B$ is a monoid homomorphism, then
  $\alpha\prn{B} \circ F\prn{f} = G\prn{f} \circ \alpha\prn{A}$.
\end{restatable*}
\noindent
In particular, choosing $F = \Con{List}$, $G = \ArrId{}$, and $\alpha = \Con{sum}$ yields our
desired earlier example.

Replacing $\Con{Monoid}$ with $\Con{Ring}$, one could derive a similar theorem to argue that given a
numerical algorithm $f : (R : \Con{Ring}) \to R^n \to R$ then the parity of its output (when applied
to $\mathbb{Z}$) depends on the parity of its inputs, as $f$ commutes with the map
$\mathbb{Z} \to \mathbb{Z}/2$. These are instances of a \emph{directed} version of the structure
identity principle, the structure homomorphism principle (SHP)~\parencite{coquand:sip:2013,hottbook,ahrens:2022,weaver:2020}: if $C$ is a type
of algebraic structures, its homomorphisms coincide with classical morphisms of those
structures. Consequently every term and type using $C$ is therefore automatically bound to be
functorial and respect these classical morphisms. It was observed by \eg{}, \textcite{weaver:2020} that
SHP could be used to ease formalization efforts and we provide the first complete examples
of this and by proving SHP occurs for a wide class of structures.

More broadly, just as \HOTT{} allowed us to internalize parametricity results based on equivalence
relations, \STT{} allows us to internalize parametricity arguments based on naturality. From this,
we can also recover a classic result:
\begin{restatable*}{lemma}{polyid}
  If $f : \prn{A : \Space} \to A \to A$ then $f = \lambda A\,a.\,a$.
\end{restatable*}
We may summarize these results by the slogan ``$\Space$ is a type which must be used
\emph{covariantly}.'' In particular, any type depending on $\Space$ (or types derived from it)
must be functorial in this argument.

Recreating parametricity arguments, however, is far from the only use of $\Space$. Just as we
defined $\Con{Monoid}$, we can define various categories critical for $\infty$-category theory, such
as the category of partial orders, the simplex category, the category of finite sets, \etc{} Using
these, we present the first steps towards formalizing \emph{higher algebra} (one of the main
applications of $\infty$-category theory) within type theory. Higher algebra is most often
encountered by type theorists in the form of the \emph{coherence problem} and, from this point of
view, using $\Space$ we are able to give definitions of infinitely coherent monoids, groups, \etc{}
Fundamentally, having just $\Space$ available throws open the door to defining a wide variety of
derived categories and all the applications this entails.

\subsection{Constructing \texorpdfstring{$\Space$}{S}} In a certain sense, the difficulty with
$\Space$ to \STT{} is not so much in its addition---we could always postulate a type $X : \Uni$
along with terms for the Segal and Rezk axioms, declare it to be $\Space$, and call it a day! The
challenge comes in finding a complete API for $\Space$ within \STT{} that, when established, allows
us to prove all expected results and determines $\Space$ up to a contractible choice of isomorphism.
This is where $\infty$-categories prove substantially more complex than $1$-categories. It no longer
suffices to specify objects and morphisms to define $\Space$, we must also specify the higher
simplices needed for coherent composition. Thus, even if we set aside the distasteful nature of
simply adding axioms to construct $\Space$, we would be left with the task of adding an
\emph{infinite} number of axioms on top of \eg{}, directed univalence to fully specify its behavior.
This is a famous problem of $\infty$-category theory where nearly all constructions must be carried
out indirectly through heavy machinery.

Our main theorem therefore is to construct $\Space$ internally and thereby provide a complete API
for its use. We do this by adapting the methods of \textcite{licata:2017,weaver:2020} to prove one of the most
widely-used results in $\infty$-category theory, the straightening--unstraightening
equivalence~\parencite{lurie:2009,heuts:2016,cisinski:2019,cisinski:2022}, inside of type theory. Roughly, we define
$\Space$ and prove that the type $X \to \Space$ is equivalent to the subtype of $X \to \Uni$ spanned
by \emph{amazingly covariant families}. That is, a map $X \to \Space$ corresponds to a type family
over $X$ which is covariant in $X$ as well as the context \ie{}, \emph{amazingly}
covariant~\parencite{riley:2024}.

We show that all the central properties of $\Space$ follow from this description. For instance, we
are able to show that $\Space$ is closed under the expected operations (limits, colimits, dependent
sums, and certain dependent products) and, most importantly, we prove the directed univalence axiom.

\subsection{Extending simplicial type theory to triangulated type theory}
The central challenge is giving an adequate definition of \emph{amazingly covariant} families: types
$\Gamma \vdash A : X \to \Uni$ which are covariant not only in $X$, but the entire context
$\Gamma$. This second condition, however, cannot be expressed inside of simplicial type
theory. Similar situations have arisen in many contexts within \HOTT{}~%
\parencite{%
  schreiber:2013,%
  schreiber:2014,%
  shulman:2018,%
  myers:2023} and, as in prior work, we address this lack of expressivity by extending simplicial
type theory by a collection of \emph{modalities} to capture amazing covariance.

In fact, even without amazing covariance we are led to modal simplicial type theory or indeed, modal
versions of any type theory seeking to internalize directed univalence. The reason why can be summed
up in a single word: contravariance. It is all well and good to have a type whose use is
automatically covariant, but common operations on the universe (\eg{}, $X \mapsto X \to \Bool$) are simply
not covariant, and some (\eg{} $X \mapsto X \to X$) are neither co- nor contravariant. As it stands,
$\Space$ can only be used covariantly and so we cannot express these important and natural
operations. To rectify this, we extend \STT{} with modalities which allow us to express
\emph{contravariant} dependence on $\Space$ as well as \emph{invariant} dependence. Both of these
modalities have central positions within synthetic category theory: the first sends a category to
its opposite and the second sends a category to its underlying groupoid of objects. While neither
operation can be realized as a function $\Uni \to \Uni$~\parencite{shulman:2018}, both of these
operations can be included as modalities~\parencite{gratzer:mtt-journal:2021}.

Having accepted that some modalities are necessary for simplicial type theory, it is then natural to
ask what other modalities must be added in order to internally define amazing covariance and
$\Space$.  Following \textcite{licata:2018}, we would like to include a modality which behaves like the
right adjoint to $A \mapsto \prn{\Int \to A}$; the so-called amazing right adjoint to $\Int \to
-$. In \opcit{}, the intended model (cubical sets) had such a modality but in the standard model of
simplicial type theory, no such right adjoint exists. Accordingly, we could add such a modality to
simplicial type theory, but we would have no means by which to justify it. In order to address this,
we must also weaken the standard model of simplicial type theory and, with it, the assumed structure
on $\Int$. Rather than postulating a totally ordered $\Int$, we only ask that $\Int$ be a bounded
distributive lattice where $0 \neq 1$. Semantically, this corresponds to shifting from simplicial
spaces---the standard model---to the larger category of cubical\footnote{Technically, we work within
  the category of Dedekind cubical spaces. See \cref{sec:model}.} spaces. Within this new
category, the necessary right adjoint exists and we can justify the addition of the necessary
modality. In order to manipulate these new modalities and relaxed interval, we also axiomatize
several general facts from the cubical spaces model. All told, we work within a version of
\MTT{}~\parencite{gratzer:mtt-journal:2021} (to account for modalities) and with a less structured
interval $\Int$. We term the result \emph{triangulated type theory} \TTT{}.

Within \TTT{}, we isolate \emph{simplicial types}, those which
\emph{believe} the interval to be totally ordered. Simplicial types ``embed'' \STT{} into \TTT{} and
we are eventually interested only in these types (in fact, mostly in simplicial Rezk types).
However, the presence of non-simplicial types is crucial to allow for the constructions needed to
define $\Space$---even though $\Space$ will itself turn out to be simplicial Rezk.

Finally, we note that while \MTT{} enjoys canonicity~\parencite{gratzer:normalization:2022}, adding
axioms (univalence, $\Int$, \etc{}) obstructs computation and so canonicity does not hold for
\TTT{}. Accordingly, \TTT{} is closer to ``book \HOTT{}''~\parencite{hottbook} than cubical type
theory~\parencite{cohen:2017}. We leave it to future work to develop computational versions of
our new axioms and integrate existing computational accounts of univalence in
\MTT{}~\parencite{aagaard:2022}.

\subsection{Contributions}
We contribute \TTT{}, a modal extension of simplicial type theory, and use it to construct a
directed univalent universe of groupoids $\Space$. In so doing, we construct the first non-trivial
examples of categories within simplicial type theory. More specifically:
\begin{itemize}
\item We identify several general and reusable reasoning principles with which to extend \STT{}.
\item We prove that $\Space$ satisfies (directed) univalence, as well as the Segal and Rezk
  conditions.
\item We construct \emph{full subcategories} purely internally and isolate important subcategories
  of $\Space$.
\item We build numerous important classical examples of categories \eg{}, presheaves, spectra,
  partial orders, and other (higher) algebraic categories from $\Space$.
\end{itemize}
Finally, we crystallize a conjectured \emph{structure homomorphism principle} which can be
used to recover various parametricity arguments as well as automatically discharge functoriality
goals and proof obligations. We give the first complete example applications of this principle.

We have endeavored throughout this paper to make most proofs reasonably explicit. This is not
only for the sake of rigor, but because a major contribution of our synthetic approach with both
\STT{} and \TTT{} is the comparative simplicity of the proofs. Crucially, no knowledge of
$\infty$-categories or the semantics of homotopy type theory is required by our key
arguments. Even the most complex arguments in \cref{sec:space} take up only half of page and are
possible to follow to those experienced with (modal) type theory. Ideally, we would substantiate
this claim by formalizing our arguments in a proof assistant, but there is presently no suitably
general implementation of modal type theory.

\begin{acks}
  For interesting and helpful discussions around the material of this work we would like to thank
  Mathieu Anel,
  Steve Awodey,
  Fredrik Bakke,
  Lars Birkedal,
  Tim Campion,
  Evan Cavallo,
  Felix Cherubini,
  Denis-Charles Cisinski,
  Bastiaan Cnossen,
  Jonas Frey,
  Rune Haugseng,
  Sina Hazratpour,
  Andr{\'e} Joyal,
  Dan Licata,
  Louis Martini,
  Hoang Kim Nguyen,
  Nima Rasekh,
  Emily Riehl,
  Maru Sarazola,
  Christian Sattler,
  Michael Shulman,
  Jonathan Sterling,
  Thomas Streicher,
  Chaitanya Leena Subramaniam,
  Paula Verdugo,
  Dominic Verity,
  Florrie Verity,
  Matthew Weaver,
  and Sebastian Wolf.

  Daniel Gratzer was supported in part by a Villum Investigator grant (no. 25804), Center for Basic
  Research in Program Verification (CPV), from the VILLUM Foundation. Ulrik Buchholtz acknowledges the support
  of the Centre for Advanced Study (CAS) at the Norwegian Academy of Science and Letters in Oslo,
  Norway, which funded and hosted the research project Homotopy Type Theory and Univalent Foundations during the academic year 2018/19, during which early stages of this work have been carried out. Jonathan Weinberger is grateful to the US Army Research Office for the support of this work under MURI Grant W911NF-20-1-0082. He also acknowledges the support of CAS and and the hospitality of Bj{\o}rn Ian Dundas and Marc Bezem on the occasion of several guest visits of the HoTT-UF project.
\end{acks}

\section{A primer on simplicial and modal type theory}
\label{sec:stt}

Before diving into the construction of the universe of groupoids, we recall some of the details of
simplicial type theory from \textcite{riehl:2017} and its modal extension. Both simplicial type
theory and the modal type theory we combine it with are extensions of homotopy type theory and so,
while we assume some familiarity with \HOTT{}, we recall some of the basic notions ``book \HOTT'' as
described by the \textcite{hottbook} to fix our chosen notation.

Recall that book \HOTT extends an ordinary type theory with the univalence axiom. For us, this
ordinary will be intensional Martin-L{\"o}f type theory with a hierarchy of universes $\Uni[0] :
\Uni[1] : \dots$ \etc{} We will further assume that these universes are cumulative and closed under
all relevant connectives. Notably, we assume our universes are closed under propositional truncation
and---in one instance---pushouts. We do not require that these higher inductive types satisfy any
particular definitional equalities, as this is presently not supported by the interpretation of
\HOTT{} into an $\infty$-topos~\parencite{lumsdaine:2020,shulman:2019}.

We will follow \textcite{hottbook} and write
$a =_A b$ (or, even more tersely) $a = b$ for the intensional identity type. Moreover, if $p : a =
b$ we shall write $p_* : B\prn{a} \to B\prn{b}$ for the transport function defined by path induction
on $p$ associated with $B : A \to \Uni$ and, on occasion, $p \bullet q$ for the concatenation of two
paths.

Most importantly, we shall assume that each universe $\Uni[i]$ satisfies Voevodsky's univalence
axiom. In particular, if we write $A \Equiv B$ for the subtype of $A \to B$ spanned by equivalences,
we assume that the following canonical map is an equivalence:
\[ \Con{ua} : \prn{A\,B : \Uni[i]} \to \prn{A = B} \to \prn{A \Equiv B} \]
We refer the reader again to \textcite{hottbook} for a thorough discussion of this axiom. Finally,
we recall a few crucial notations from \opcit{} which we shall repeatedly use:
\begin{gather*}
  \IsContr,\IsProp,\IsSet : \Uni \to \Uni;
  \quad \IsContr\,A = \Sum{a : A} \Prod{b : B} a = b,
  \\
  \IsProp\,A = \Prod{a\,b : A} \IsContr\prn{a = b},
  \quad
  \IsSet\,A = \Prod{a\,b : A} \IsProp\prn{a = b}
\end{gather*}

These predicates respectively isolate (1) types which behave like $\Unit$, \ie{} are contractible,
(2) types which behave like propositions, and (3) types which behave like discrete spaces (\ie{}
sets). In homotopical parlance, these are the $(-2)$-, $(-1)$-, and $0$-truncated types. In fact, we can
define each of these predicates as instances of a more general $\Con{hasHLevel} : \Nat \to \Uni \to
\Prop$, but we do not have need for this additional generality. Each of these induce subtypes of the
universe \eg{}, $\Prop = \Sum{A : \Uni} \IsProp\,A$. For instance, we may speak of families of
propositions over $A$ (predicates) using maps $A \to \Prop$.

\subsection{Simplicial type theory and basic category theory}

We now turn to simplicial type theory, an extension of \HOTT{} designed to reason about simplicial
spaces and, through them, $\infty$-categories. The main axiom of simplicial type theory 
asserts the existence of a type which internalizes the representable $\Delta^1$ or,
equivalently, the category with two objects and one non-trivial morphism connecting them: 
\begin{definition}
  Core simplicial type theory \STT{} extends homotopy type theory with the following:
  \begin{enumerate}
  \item A \emph{directed interval} type $\Int : \HSet$
  \item The operations and equations shaping $\Int$ into a bounded total order $\prn{0,1,\le}$.
  \end{enumerate}
\end{definition}

The precise form of this axiom is subject to numerous variations. For instance, in the original
incarnation of simplicial type theory~\parencite{riehl:2017} featured a bespoke judgmental structure for
$\Int$ which enabled additional definitional equalities. In order to keep the system simple in
anticipation of adding various modalities to it presently, we have eschewed this structure so that
$\Int$ is an ordinary type. Moreover, since we shall shortly be interested in a model of simplicial
type theory in \emph{cubical} spaces, we have only required that $\Int$ be a bounded distributive
lattice rather than a linear order.

Using the lattice structure on $\Int$, we can now specify the common simplicial shapes used to model
composition in $\infty$-category theory \eg{} $\Delta^n$:
\[
  \Delta^0 \defeq \Unit
  \quad
  \Delta^{n+1} \defeq
  \Compr{\prn{i_1, \ldots, i_{n+1}} : \Int^{n+1}}{i_1 \ge i_2 \geq \ldots \geq i_{n+1}}
  \quad
  \Lambda^2_1 \defeq \Compr{\prn{i,j} : \Int^2}{i = 1 \lor j = 0}
\]
One can also give general descriptions of the boundaries $\partial \Delta^n$ and the $(n,k)$-horns
$\Lambda_k^n$, for $n \geq 0$ and $0 \leq k \leq n$~\cite[Section~3]{riehl:2017}. We use these to
define categories and related structures.

\begin{definition}
  Given $a,b : A$, the type of \emph{homomorphisms} or \emph{arrows} from $a$ to $b$ is given by
  \[
    \Hom[A]{a}{b} \defeq \Sum{f : \Int \to A} f\,0 = a \times f\,1 = b
  \]
  In other words, $\Hom[A]{a}{b}$ is the (homotopy) fiber of $A^\Int \to A \times A$ over
  $\prn{a,b}$. For convenience, we suppress the forgetful map $\Hom[A]{a}{b} \to \prn{\Int \to A}$
  and so will write $f\prn{i}$ when $f : \hom_A(a,b)$ or similar.
\end{definition}

\begin{notation}
  We write $\GenArr$ for the tautological homomorphism $\Hom{0}{1}$ in $\Int$ induced by $\ArrId{}$.
\end{notation}

We can relativize the notion of homomorphisms to dependent types:
\begin{definition}
  Given $a,b : A$ and $f : \Hom[A]{a}{b}$, for a type family $P : A \to \Uni$, a \emph{dependent
    homomorphism} from $x : P\,a$ and $y : P\,b$ over $f$ is given by
  \[
    \Hom[P]<f>{x}{y}
    \defeq
    \Sum{\varphi : \prn{i : \Int} \to P\prn{f\,i}}
    (\prn{\Proj[2]\,{f}}_*\prn{\varphi\,0} = x)
    \times
    (\prn{\Proj[3]\,{f}}_*\prn{\varphi\,1}= y)
  \]
  Note that we must transport by the identifications $\Proj[2]{f} : f\prn{0} = a$ and
  $\Proj[3]{f} : f\prn{0} = b$ in order to ensure that these equalities are type-correct.%
  \footnote{These transports are precisely what the judgmental extensions of \textcite{riehl:2017}
    aim to avoid. Fortunately they do not occur so frequently as to be a major impediment in this
    paper.}
\end{definition}

With the apparatus of morphisms to hand, we may recall the following definition of
\emph{pre-categories} \ie{} \emph{Segal types} from the introduction:

\begin{definition}
  $A : \Uni$ is \emph{Segal} if the canonical map
  $i : (\Delta^2 \to A) \to (\Lambda_1^2 \to A)$ is an equivalence.
\end{definition}

\begin{notation}
  If $A$ is Segal and $f : \Hom{a}{b}$, $g : \Hom{b}{c}$, we write $g \circ f$ for the map $\Int \to
  A$ given by $\lambda t.\,\prn{i^{-1}\prn{f,g}}\,\prn{t,t}$ \ie{} the long edge of the triangle obtained
  by extending $\prn{f,g} : \Lambda^2_1 \to A$ to $\Delta^2 \to A$. This operation is automatically
  associative and constant functions $\Int \to A$ (identity homomorphisms) are units for $\circ$. A
  major benefit of working in simplicial type theory is that such a composition does exist, even
  though composites are defined only up to a contractible choice.
\end{notation}

\begin{definition}
  We say an arrow $f : \Hom{a}{b}$ in a Segal type $A$ is an \emph{isomorphism} if the following
  type is inhabited: $\IsIso(f) \defeq \Sum{g\,h : \Hom{b}{a}} (g \circ f = \ArrId{a}) \times (f
  \circ h = \ArrId{b})$.
\end{definition}
Note that $\IsIso(f)$ is a proposition and we denote the induced subtype of $\Hom{a}{b}$ by
$a \cong_A b$. With the definition of isomorphism to hand, we can properly define \emph{categories}
and \emph{groupoids}:

\begin{definition}
  If $A$ is Segal, we say $A$ is a \emph{category}/\emph{Rezk-complete} if the following map
  (defined by path induction) is an equivalence: $\IdToIso : \prn{a,b:A} \to \prn{a = b} \to \prn{a \cong b}$
\end{definition}

\begin{remark}
  We note that the proposition $\IsIso\prn{f}$ requires that $f$ have a section and a retraction; a common
  definition of equivalence in \HOTT{}. We shall return to this point in \cref{sec:space} when we
  prove that our directed-univalent universe satisfies the Rezk condition.
\end{remark}

\begin{definition}
  A type $A$ is \emph{groupoid} or a \emph{space} or \emph{$\Int$-null} if
  $\prn{a = b} \to \prn{a \to b}$ is an equivalence.\footnote{The terminology ``$\Int$-null'' stems from
    \textcite{rijke:2020}; it is equivalent to requiring that the constant map $A \to \prn{\Int \to
  A}$ is an equivalence.}
\end{definition}

\begin{lemma}
  A type is a groupoid if and only if it is a category where every arrow is an isomorphism.
\end{lemma}

Intuitively, a type is Rezk when it satisfies a kind of univalence condition: isomorphism is
identity. In the intended model of \TTT{}, they correspond to complete Segal spaces, in turn, model
$(\infty,1)$-categories. Op.~cit.~further show that maps between Segal types are
automatically \emph{functors} \ie{} they preserve composition and identities.

\subsection{Multimodal type theory}
\label{sec:stt:mtt}

As mentioned in \cref{sec:introduction}, we must extend type theory with various \emph{modalities}
in order to define $\Space$. We shall do this by ``rebasing'' simplicial type theory atop
\MTT{}~\parencite{gratzer:mtt-journal:2021}, a general framework for modal type theory. In
particular, we shall take our base type theory to not just be intensional Martin-L{\"o}f type
theory, but a slightly richer theory which includes various modalities alongside the ordinary
constructors of dependent type theory. Since \MTT{} is already a complex type theory without any of
the additional axioms of homotopy or simplicial type theory, we give a brief overview of the theory
now. We refer the reader to \textcite{gratzer:mtt-journal:2021} or \textcite[Chapter 6]{gratzer:phd}
for a more thorough introduction to the theory. We will also explain \MTT{} as it is used in this
paper: in the same informal style that we will use type theory generally. Accordingly, we will not
focus overmuch on the substitution calculus of \MTT{} or other aspects of its metatheory which need
not concern us presently. An exception to this general pattern is the description of the semantics
of \TTT{}, but we will briefly recall the model theory of \MTT{} at that point.

We must immediately note that \MTT{} is not, properly speaking, a type theory. Rather, \MTT{} is a
\emph{framework} for modal type theories: a user picks a mode theory---a 2-category describing their
modalities---and \MTT{} produces a type theory for working with this collection of modalities. For
this exposition, let us fix $\Mode$ an arbitrary 2-category where we think of objects (\emph{modes})
$m,n$ as type theories connected by the 1-cells (\emph{modalities}) $\mu,\nu$. The 2-cells
$\alpha,\beta$ encode transformations between modalities enabling us to control \eg{}, whether $\mu$
is a comonad. In fact, for our particular use case we shall have at most one 2-cell between any pair
of modalities and exactly one mode. For simplicity, we shall assume the same to be true of $\Mode$
and we write $\mu \le \nu$ when there exists a (necessarily unique) 2-cell from $\mu$ to $\nu$.
Instantiating \MTT{} with $\Mode$ yields a type theory which includes a modal type for each $\mu$ in
the mode theory and, furthermore, these modal types are \emph{2-functorial}.

The basic modification \MTT{} makes to the type theory is to change the form of variables in the
context. A context is no longer simply a telescope of bindings $x : A$. Instead, each declaration is
annotated by a pair of modalities $\DeclVar{x}{\mu/\nu}{A}$. The annotation $\mu/\nu$ signifies that
$x$ was constructed under the $\mu$ modality and, presently, we are working to
construct an element of the $\nu$ modality.

\begin{notation}
  In a declaration $\DeclVar{x}{\mu/\nu}{A}$ we shall often omit $\mu$ or $\nu$ if they are the
  identity \eg{}, $\DeclVar{x}{\mu}{A}$ or $y : B$ rather than $\DeclVar{x}{\mu/\ArrId{}}$ or
  $\DeclVar{y}{\ArrId{}/\ArrId{}}{B}$.
\end{notation}

Both halves of the annotation $\mu/\nu$ restrict how variables are used to prevent terms from illegally escaping or entering modalities and, roughly, we are
allowed to use a variable when they cancel.

\begin{mathpar}
  \inferrule{
    \mu \le \nu
    \\
    \DeclVar{x}{\mu/\nu}{A} \in \Gamma
  }{
    \IsTm[\Gamma]{x}{A}
  }
  \and
  \inferrule{
    \IsTm[\Gamma/\mu]{a}{A}
    \\
    \IsTm[\Gamma, \DeclVar{x}{\mu/\ArrId{}}{A}]{b\prn{x}}{B\prn{x}}
  }{
    \IsTm[\Gamma]{\Sb{b}{a/x}}{\Sb{B}{b/x}}
  }
\end{mathpar}

In the above, $\Gamma/\mu$ denotes the context with the same variables as $\Gamma$ but where
$\DeclVar{x}{\nu/\nu_0}{A}$ is replaced by $\DeclVar{x}{\nu/\prn{\nu_0 \circ \mu}}{A}$. Note we have
presented only the relevant and simpler substitution rule allowing us to discharge an assumption
with the ``denominator'' of an annotation is the identity. Note also that \MTT{} does not alter the
actual definition of substitution from type theory---one merely proves after the fact that ordinary
substitution does not result in terms using inaccessible variables.

These annotations are also used to introduce the modal types associated with each $\mu$.
For instance, suppose we have a modality $\mu$, and we intend to form the modal type
$\Modify[\mu]{A}$. This is well-formed in context $\Gamma$ just when $A$ is well-formed in the
context $\Gamma/\mu$. Similarly, we can form an element of the modal type
$\MkMod[\mu]{a} : \Modify[\mu]{A}$ in context $\Gamma$ just when $a : A$ in the context
$\Gamma/\mu$:
\begin{mathpar}
  \inferrule{
    \IsTy[\Gamma/\mu]{A}
  }{
    \IsTy[\Gamma]{\Modify[\mu]{A}}
  }
  \and
  \inferrule{
    \IsTm[\Gamma/\mu]{a}{A}
  }{
    \IsTm{\MkMod[\mu]{a}}{\Modify[\mu]{A}}
  }
\end{mathpar}

The elimination rule for $\Modify[\mu]{-}$ papers over the difference between
$\DeclVar{a}{\nu\circ\mu/\ArrId{}}{A}$ and
$\DeclVar{a'}{\nu/\ArrId{}}{\Modify[\mu]{A}}$:
\begin{mathparpagebreakable}
  \inferrule{
    \IsTy[\Gamma/\nu\circ\mu]{A}
    \\
    \IsTy[\Gamma, \DeclVar{y}{\nu/\ArrId{}}{\Modify[\mu]{A}}]{B\prn{y}}
    \\\\
    \IsTm[\Gamma, \DeclVar{x}{\nu\circ\mu/\ArrId{}}{A}]{b\prn{x}}{\Sb{B}{\MkMod[\mu]{x}/y}}
    \\
    \IsTm[\Gamma/\nu]{a}{\Modify[\mu]{A}}
  }{
    \IsTm{\LetMod<\nu>[\mu]{a}[x]{b\prn{x}}}{\Sb{B}{a/y}}
  }
  \and
  \LetMod<\nu>[\mu]{\MkMod{a_0}}[x]{b\prn{x}} = \Sb{b}{a_0/x}
\end{mathparpagebreakable}

In particular, if we are attempting to construct a term using a variable $x$ of type $\Modify{A}$,
we may reduce to the case where $x = \MkMod{x_0}$ for a fresh variable $\DeclVar{x_0}{\mu}{A}$. More
technically, this amounts to a certain canonical map being weakly orthogonal to all types. We shall
revisit this perspective in the discussion of semantics.

Already, these rules are sufficient to prove the following facts that we shall use ubiquitously:
\begin{lemma}
  If $\IsTy[\Gamma/\nu\circ\mu]{A}$ then
  $\Modify[\nu]{\Modify[\mu]{A}} \Equiv \Modify[\nu\circ\mu]{A}$ and if\,$\IsTy{B}$ then
  $\Modify[\ArrId{}]{B} \Equiv B$.
\end{lemma}

We have already seen that $\mu \le \nu$ allows us to access variables under $\mu/\nu$. This,
combined with the elimination rule for modal types, allows us to produce a function introduces a
function $\Modify[\mu]{-} \to \Modify[\nu]{-}$. In order to make this well-formed, we note the
following admissible principle in \MTT{}, akin to the admissibility of weakening in ordinary type
theory:
\begin{lemma}
  If $\IsTm[\Gamma/\nu]{a}{A}$ and $\mu \le \nu$ then $\IsTm[\Gamma/\mu]{a}{A}$.
\end{lemma}

\begin{lemma}
  If $\Mor[\alpha]{\mu}{\nu}$ and $\IsTy[\Gamma/\mu]{A}$ then there is a map $\Coe^{\mu \le \nu} :
  \Modify[\mu]{A} \to \Modify[\nu]{A}$.
\end{lemma}
\begin{proof}
  Though this is an elementary result, we give a proof to highlight the process of working in
  \MTT{}. Suppose we are given $x : \Modify[\mu]{A}$, we must construct a term of type
  $\Modify[\nu]{A}$. Using the elimination principle for $\Modify{A}$, we may fix
  $\DeclVar{x_0}{\mu}{A}$ and assume that $x = \MkMod{x_0}$. Next, using the introduction rule for
  $\Modify[\nu]{A}$, it suffices to construct an element of $A$, though we must update the
  annotations on $x$ and $x_0$ to $\DeclVar{x_0}{\mu/\nu}{A}$ and
  $\DeclVar{x}{\ArrId/\nu}{\Modify{A}}$. Finally, we may use the variable rule to access $x_0 : A$
  as required. All told then, the full term is given as follows:
  \[
    \Coe^{\mu \le \nu} = \Lam[x]{\LetMod{x}[x_0]{\MkMod[\nu]{x_0}}}
  \]
  This pattern of binding a variable of type $\Modify{A}$ and immediately pattern-matching upon it
  is exceptionally common. Accordingly, we adopt the following ``pattern-matching'' style notation
  for convenience:
  \[
    \Coe^{\mu \le \nu}\prn{\MkMod{x_0}} = \MkMod[\nu]{x_0} \qedhere
  \]
\end{proof}

\begin{notation}
  \label{not:modfun}
We shall also have occasion to use the convenience feature of modalized dependent products
$\prn{\DeclVar{a}{\mu}{A}} \to B\prn{a}$ which abstract over $\DeclVar{a}{\mu}{A}$ directly rather
than $\DeclVar{a}{\ArrId{}}{\Modify{A}}$ to allow us to avoid immediately pattern-matching on $a$.
In particular, $\prn{\DeclVar{a}{\mu}{A}} \to B\prn{a}$ is equivalent to
$\prn{\DeclVar{a}{\ArrId{}}{\Modify{A}}} \to \prn{\LetMod{a}[a_0]{B\prn{a_0}}}$. We shall most often
use this when working informally to type theory. For instance, if we write ``given
$\DeclVar{a}{\mu}{A}$, there exists an element of $B\prn{a}$'' this should be interpreted as
denoting the type $\prn{\DeclVar{a}{\mu}{A}} \to B\prn{a}$.
\end{notation}

\section{Triangulated type theory}
\label{sec:ttt}

As already mentioned, part of our strategy is to replicate the argument of \textcite{weaver:2020} in \STT{} combined with
\MTT{} and build a directed univalent universe of groupoids. However, there is a fundamental problem
with this approach: \textcite{weaver:2020} rely on a particular modality (the right adjoint to $\Int
\to -$) in their construction, and the standard model of simplicial type theory in $\PSH{\SIMP}$
simply does not have an analog to this modality.  Thus, we need more than a combination of \STT{}
and \MTT{}, we need a new system which admits a model where $\Int \to -$ is a left adjoint.

To motivate our solution, let us further consider the intended model of \STT{} in simplicial spaces
for a moment. Up to the complexity needed to model homotopy type theory, these are simplicial sets
\ie{}, presheaves on $\SIMP$, the category of finite, inhabited total linear orders. The interval
$\Int$ is realized by $\Yo{\brk{1}}$. Our problem then amounts to the fact that exponentiation by
$\Yo{\brk{1}}$ does not have a right adjoint.\footnote{This is easiest to check by observing that it
does not commute with pushouts.} However, there is a category closely related to $\SIMP$ which has
also received a great deal of attention by type theorists interested in cubical type theory: the
category of (Dedekind) cubes $\CUBE$, the full subcategory of partial orders spanned by $\brc{0 \le
1}^n$ for all $n$. The category $\CUBE$ enjoys two properties which make it interesting for simplicial type
theorists: there is a fully faithful functor $\PSH{\SIMP} \to \PSH{\CUBE}$ which preserves the
interval and admits both left and right
adjoints~\parencite{sattler:2019,kapulkin:2020,streicher:2021}, and within $\PSH{\CUBE}$
exponentiating by the interval is a left adjoint.

Accordingly, we introduce a relaxation of simplicial type theory intended to capture (the
homotopical version of) $\PSH{\CUBE}$. Within this type theory, we can recover simplicial type
theory by studying those types which are in the image of the aforementioned embedding alongside the
amazing right adjoint necessary for constructing our sought-after universe~\parencite{licata:2018}.

Concretely, we work within a version of \MTT{} instantiated with several modalities, further
extended by a bounded distributive lattice $\Int : \HSet$ which serves as our weakened version of
the interval, and a handful of axioms. Notably, we no longer assume that $\Int$ is totally ordered
and instead ask for it to be a bounded distributive lattice; this is our central deviation from
simplicial type theory.

In this section, we introduce \emph{triangulated type theory} as an extension of \MTT{}. For
precision, we give a complete listing of the axioms we require (including univalence, the existence
of an interval, \etc{}). Finally, we explore a few elementary consequences of this axiom and produce
our first non-trivial examples of categories.

\subsection{The definition of triangulated type theory}
\label{sec:ttt:axioms}

We begin by describing the particular instantiation of \MTT{} needed for \TTT{}.

\subsubsection{The mode theory}
As mentioned in \cref{sec:stt:mtt}, \MTT{} must be instantiated by a mode theory. In our case, we
shall require only one mode $m$ which we shall think of as cubical spaces $\PSH[\SSET]{\CUBE}$. We
shall then add the following modalities
\begin{itemize}
\item A pair of modalities $\GM,\SM$ internalizing the global sections comonad and its right adjoint.
\item A modality $\OM$ internalizing the ``opposite.''
\end{itemize}

Intuitively, $\Modify[\OM]{X}$ is a type with the same points as $X$, but with all the higher cubes
reversed; if there was a line joining $x_0$ to $x_1$ in $X$, then $\Modify[\OM]{X}$ will have a line
joining $x_1$ to $x_0$ instead.

On the other hand, $\Modify[\GM]{-}$ deletes all (higher) cells from a type, leaving only the
underlying groupoid of points (its \emph{groupoid core}). We shall use this modality to define the
core of a category and, more generally, use it to isolate discrete categories. Owing to this second
point, prior work has often referred to $\GM$ as the \emph{discrete}
modality~\parencite{shulman:2018,myers:2023}. Its right adjoint, $\Modify[\SM]{-}$, is slightly less
intuitive. Operationally, it deletes all higher cells from a type and then adds in a unique (higher)
cell between every collection of points. We caution the reader that while $\Modify[\GM]{X}$ is
always a groupoid in the sense of \cref{sec:stt}, it is often the case that $\Modify[\SM]{X}$ is
not a category even if $X$ was originally a category. This is not unexpected: $\Modify[\GM]{-}$
models the core functor sending a category to its underlying groupoid, but this operation when
restricted to categories has no right adjoint. It is only in a bigger category such as cubical or simplicial spaces
that $\Modify[\SM]{-}$ exists.

We require a number of equations and inequalities to force these modalities to behave as expected.
In particular, we require the following 2-cells and equations on modalities:
\begin{mathpar}
  \GM \circ \GM = \GM \circ \SM = \GM \circ \OM = \OM \circ \GM  = \GM
  \and
  \SM \circ \SM = \SM \circ \GM = \SM \circ \OM = \OM \circ \SM = \SM
  \\
  \OM \circ \OM = \ArrId{}
  \and
  \GM \le \ArrId{}
  \and
  \ArrId{} \le \SM
\end{mathpar}

We refer to the 2-category theory generated by these constraints as $\TTTMode$.

\subsubsection{The interval}
As mentioned previously, we require an interval in order to capture the simplicial (or, in our case,
cubical) structure.

\begin{axiom}[The interval]
  \label{ax:int}
  There exists a bounded distributive lattice $\prn{\Int : \HSet,\land,\lor,0,1}$.
\end{axiom}

Our next axiom controls the behavior of the opposite modality on $\Int$:
\begin{axiom}[Opposite of $\Int$]
  \label{ax:op-of-int-is-int}
  There is an equivalence $\neg : \Modify[\OM]{\Int} \to \Int$ which swaps $0$ for $1$ and $\lor$ for $\land$.
\end{axiom}
\begin{notation}
  In various places, it will be convenient to treat $\neg$ as a function $\MFn[\OM]{\Int}{\Int}$ to
  avoid spuriously introducing $\MkMod[\OM]{-}$. The two types $\MFn[\OM]{\Int}{\Int}$ and
  $\Modify[\OM]{\Int} \to \Int$ are canonically equivalent and so this causes no ambiguity, see~\Cref{not:modfun}.
\end{notation}

With $\Int$, we are now able to postulate an amazing right adjoint operation to $\Int \to -$. As was
noted in \textcite{licata:2018}, this operation cannot be defined as a map $\Uni \to \Uni$.
Accordingly, we restrict its action to global elements using $\GM$.

\begin{axiom}[$\Int$ is tiny]
  \label{ax:int-is-left-adjoint}
  The following proposition holds:
  \[
    \prn{\DeclVar{A}{\GM}{\Uni}} \to
    \Sum{\DeclVar{A_\Int}{\GM}{\Uni}} \Sum{\DeclVar{\epsilon}{\GM}{\prn{A_\Int}^\Int \to A}}
    \Prod{\DeclVar{B}{\GM}{\Uni}} \IsEquiv\prn{\Modify[\GM]{B \to A_\Int} \to \Modify[\GM]{B^\Int \to A}}
  \]
  In other words, we require that for each element $\DeclVar{A}{\GM}{\Uni}$, there exists a type
  $A_\Int$ which represents the functor $\prn{-}^\Int \to A$ \ie{}, a right adjoint to $\Int \to -$.
\end{axiom}

\subsubsection{The simplicial monad}

Before moving on to the list of additional axioms that form \TTT{}, we must take a moment to discuss
an additional construct: the simplicial monad. As motivation, while we have already noted that the
interval is not totally ordered, there is a large number of types which ``act as though it
is.'' The simplicial monad isolates and classifies these types.

More precisely, a type is \emph{simplicial} if it satisfies the following predicate:
\[
  \IsSimp\prn{A} =
  \prn{i\,j : \Int} \to
  \IsEquiv\prn{\lambda a\,z.\,a : A \to \prn{i \le j \lor j \le i \to A}}
\]
If a type $A$ satisfies $\IsSimp$, this acts as a license to totally order elements of the interval
whenever we are constructing an element of $A$. Furthermore, as the name suggests, simplicial types
are those which come from simplicial rather than cubical sets (see \cref{sec:model}).

\begin{proposition}[\textcite{rijke:2020}]
  There is a monad $\prn{\Simp : \Uni \to \Uni,\eta,\mu}$ such that:
  \begin{itemize}
  \item For every $A : \Uni$, $\IsSimp\prn{\Simp A}$ holds.
  \item If $B$ is simplicial, then $\eta^* : \prn{\Simp A \to B} \to \prn{A \to B}$ is an
    equivalence.
  \item $\Simp$ commutes with dependent sums and the identity type.
  \end{itemize}
\end{proposition}

We refer to $\Simp$ as the \emph{simplicial monad}\footnote{The notation $\Simp$ is chosen
  deliberately: simplicial types are those which believe the square $\Int \times \Int$ (along with
  all hypercubes) comes from gluing together a pair of triangles
  $\Delta^2 \Pushout{\Int} \Delta^2$.} and write $\Uni[\Simp]$ for the subtype
$\Sum{A : \Uni} \IsSimp\prn{A}$.

\begin{convention}
  We reserve the words ``category'' and ``groupoid'' for types which are simplicial in addition to
  satisfying the Segal/Rezk conditions from \STT{}. Accordingly, \eg{} \emph{category} signifies a
  type which is simplicial, Segal, and Rezk complete.
\end{convention}

\subsubsection{Additional axioms}

Finally, we require a handful of additional axioms which either improve the behavior of modalities
generally or form a more tight correspondence between our system and our intended model. We offer
some intuition for each axiom and note that each is validated by the intended model described in
\cref{sec:model}.

Our first two axioms are general and common assumptions in univalent modal type theory.
First, we record the univalence axiom here as previously discussed in \cref{sec:stt}:
\begin{axiom}[Univalence]
  \label{ax:univalence}
  We assume that each universe $\Uni[i]$ is univalent.
\end{axiom}

Next, we assume that each modality $\Modify[\mu]{-}$ commutes with identity types.
\begin{axiom}[Crisp induction]
  \label{ax:crisp-id-induction}
  For every $\mu$, the canonical map $\MkMod{a} = \MkMod{b} \to \Modify{a = b}$ is an
  equivalence.
\end{axiom}

\begin{remark}
  It is open whether \cref{ax:univalence} implies \cref{ax:crisp-id-induction}. However, in all sensible models,
  \cref{ax:crisp-id-induction} does hold and its failure to do so is more indicative of the poor
  behavior of the intensional identity type than anything else.
\end{remark}

After these fairly general reasoning principles, we now have a sequence of more
simplicial-specific axioms. The first of these links the global sections modality to the
interval. In particular, it states that the global sections of a type always form a groupoid.
\begin{axiom}[$\Int$ detects discreteness]
  \label{ax:discrete-iff-crisp}
  If $\DeclVar{A}{\GM}{\Uni}$ then $\Modify[\GM]{A} \to A$ is an equivalence ($A$ is
  \emph{discrete}) if and only if $A \to \prn{\Int \to A}$ is an equivalence ($A$ is $\Int$-null).
\end{axiom}

The next axiom states that the global points of $\Int$ itself are just $0$ and $1$ and that $0 \neq 1$:
\begin{axiom}[Global points of $\Int$]
  \label{ax:global-points}
  The canonical map $\Bool \to \Int$ is injective and $\Bool \Equiv \Modify[\GM]{\Int}$.
\end{axiom}

In our intended model, various properties can be proven by ``testing'' them at the representable
presheaves $\Yo{\brc{0 \le 1}^n}$. We include a version of this idea as an axiom in our
theory. Namely, we assert that maps between global types can be tested for invertibility at
$\Int^n$:
\begin{axiom}[Cubes separate]
  \label{ax:cubes-separate}
  A map $\DeclVar{f}{\GM}{A \to B}$ is an equivalence if and only if the following holds:
  \[
    \prn{\DeclVar{n}{\GM}{\Nat}} \to \IsEquiv\prn{f_* : \Modify[\GM]{\Int^n \to A} \to \Modify[\GM]{\Int^n \to B}}
  \]
\end{axiom}
\noindent
This follows from another possible axiom, \emph{cubes detect continuity},
following~\textcite{myers:2023}. Note that if $A$ and $B$ are simplicial, one can derive a version of
\cref{ax:cubes-separate} which replaces $\Int^n$ with $\Delta^n$.

It is relatively easy to characterize maps out of $\Simp A$ as they are closely related to maps out
of $A$ itself. It is much harder, however, to characterize $X \to \Simp A$. Our next axiom states
that in certain favorable cases these, too, coincide with the corresponding situation for $A$:
\begin{axiom}[Simplicial stability]
  \label{ax:simplicial-stability}
  If $\DeclVar{A}{\GM}{\Uni}$ then the following map is an equivalence for all $\DeclVar{n}{\GM}{\Nat}$:
  \[
    \eta_* : \Modify[\GM]{\Delta^n \to A} \to \Modify[\GM]{\Delta^n \to \Simp A}
  \]
\end{axiom}

Finally, while simplicial type theory allows us to prove many interesting facts about maps out of
the interval, it is far more difficult to prove properties about $X \to \Int$. In order to balance
the scales, we follow \textcite{cherubini:2023} and add a duality
axiom~\parencite{kock:2014,blechschmidt:2023} characterizing these maps in certain special cases.
Prior to stating this principle, we require the following definition:
\begin{definition}
  A map $\Int \to A$ of bounded distributive lattices is a \emph{finitely presented (fp)
  $\Int$-algebra} if it is merely equivalent to the canonical map $\Int \to \Int\brk{x_1, \dots,
  x_n}/\gl{t_1 = s_1, \dots, t_m = s_m}$ for some $n,m$.
\end{definition}

The definition of a homomorphism of bounded distributive lattices (a map which commutes with
$0,1,\land,\lor$) extends to a notion of homomorphism between fp $\Int$-algebras $\Hom[\Int]{A}{B}$
by further requiring the underlying map to commute with the maps $\Int \to A$ and $\Int \to B$.

\begin{axiom}[Duality]
  \label{ax:sqc}
  Given an fp $\Int$-algebra $f : \Int \to A$ the following map is an equivalence:
  \[
    \lambda a\,g.\,g\prn{a} : A \to \prn{{\Hom[\Int]{A}{\Int}} \to \Int}
  \]
\end{axiom}

\begin{definition}
  Triangulated type theory \TTT{} is \MTT{} with mode theory $\TTTMode$ extended by
  \cref{%
    ax:int,%
    ax:int-is-left-adjoint,%
    ax:op-of-int-is-int,%
    ax:univalence,%
    ax:crisp-id-induction,%
    ax:discrete-iff-crisp,%
    ax:global-points,%
    ax:cubes-separate,%
    ax:simplicial-stability,%
    ax:sqc}.
\end{definition}

\subsection{Duality and \texorpdfstring{$\Delta^n$}{Delta\{n\}}}

\cref{ax:sqc} has a number of remarkable consequences for $\Int$. While these are not specific to
directed univalent universes, they allow us to construct the first non-trivial categories inside
\TTT{}. We begin with the following result---independently proven by
\textcite{pugh:2025}---reminiscent of various principles from synthetic domain theory.

\begin{lemma}[Phoa's principle]
  \label{lem:ttt:phoa}
  Evaluation at $0,1$ is an embedding $\prn{\Int \to \Int} \to \Int \times \Int$ with image
  $\Delta^2$.
\end{lemma}
\begin{proof}
  We first will argue via \cref{ax:sqc} that $\Int\brk{x}$ is equivalent to $\Int \to \Int$ via the
  evaluation map. To see this, let us note that $\Int \to \Int\brk{x}$ is an $\Int$-algebra by
  definition, and $\Hom[\Int]{\Int\brk{x}}{\Int} \Equiv \Int$. Accordingly, by \cref{ax:sqc}, the
  map $\Con{eval} : \Int\brk{x} \to \prn{\Int \to \Int}$ is an equivalence.

  By the 2-for-3 principle of equivalences, it then suffices to show that evaluating a polynomial at
  $0$ and $1$ induces an embedding $\Int\brk{x} \to \Int \times \Int$ whose image is $\Delta^2$. An
  inductive argument allows us to conclude that $\Con{eval}\prn{p,-}$ is a monotone map from
  $\Int \to \Int$ and so evaluation of polynomials at endpoints factors through $\Delta^2$. We
  therefore are reduced to showing that this map is an equivalence.  To see this, we observe that
  any polynomial in one variable can be placed in the following normal form:
  $p = \Con{eval}\prn{p,0} \lor x \land \Con{eval}\prn{p,1}$ whereby the conclusion is immediate.
\end{proof}

\begin{notation}
  In light of the equivalence used in the proof of Phoa's principle, we will no longer distinguish
  between polynomials in one variable $\Int\brk{x}$ and functions $\Int \to \Int$.
\end{notation}

\begin{lemma}[Generalized Phoa's principle]
  \label{lem:ttt:generalized-phoa}
  \leavevmode
  \begin{itemize}
  \item The evaluation map from $\Int^n \to \Int$ to monotone maps $\Bool^n \to \Int$ is an
    equivalence.
  \item The evaluation map from $\Delta^n \to \Int$ to monotone maps $\brk{0 \le \dots \le n} \to \Int$ is an
    equivalence.
  \end{itemize}
  In the above, we have regarded $\Bool$ as a 2-element partial order $\False \le \True$.
\end{lemma}

Both claims follow from induction on $n$ and repeated application of Phoa's principle.

\begin{remark}
  The particular cube category used in our intended model of \TTT{} is equivalent by Birkhoff
  duality to the category of \emph{flat} finite bounded distributive
  lattices~\parencite{spitters:2016}. \cref{lem:ttt:generalized-phoa} is a manifestation of this
  fact.
\end{remark}

\begin{lemma}
  $\Int$ is simplicial.
\end{lemma}
\begin{proof}
  To show that $\Int \to \prn{\prn{i\le j\lor j\le i} \to \Int}$ is an equivalence, it suffices,
  by~\cref{ax:cubes-separate}, to consider $\DeclVar{f,g}{\GM}{\Int^n \to \Int}$ and show that the
  following is an equivalence:
  \[
    \Modify[\GM]{\Int^n \to \Int} \to
    \Modify[\GM]{
      \Compr{\vec x:\Int^n}{f(\vec x) \le g(\vec x) \lor g(\vec x) \le f(\vec x)} \to \Int
    }
  \]
  Using \cref{lem:ttt:generalized-phoa}, we can extend an element of the codomain to a total
  function $\Int^n \to \Int$ provided we can specify its behavior on $\vec x : \Bool^n$. The
  proposition $f(\vec x)\le g(\vec x) \lor g(\vec x)\le f(\vec x)$ holds for all $\vec x : \Bool^n$,
  and so such an extension always exists and is necessarily unique.
\end{proof}
\begin{remark}
  A more elegant proof of the above was recently provided by \textcite{williams:2025}.
\end{remark}

\begin{corollary}
  \label{cor:ttt:delta-is-category}
  $\Delta^n$ is a category.
\end{corollary}
\begin{proof}
  Since there are no nontrivial invertible morphisms in $\Delta^n$, it is trivially Rezk-complete and, as a
  retract of $\Int^n$, it is simplicial. Therefore, it suffices to show that $\Delta^n$ is Segal.

  To this end, let us consider $\Lambda^2_1 \to \Delta^n$. This is equivalent to a pair of maps
  $f,g : \Int \to \Delta^n$ such that $f\prn{1} = g\prn{0}$. Next, by the Phoa principle
  $f,g : \Int \to \Delta^n$ are fully determined by $n$-tuples of pairs \eg{},
  $\prn{\Proj[k]\prn{f\prn{0}} \le \Proj[k]\prn{f\prn{1}}}_{k \le n}$. In total then, we are given
  $n$-many 3-tuples:
  \[
    \prn{\Proj[k]\prn{f\prn{0}} \le \Proj[k]\prn{f\prn{1}} = \Proj[k]\prn{g\prn{0}} \le \Proj[k]\prn{g\prn{1}}}_{k \le n}
  \]
  By \cref{lem:ttt:generalized-phoa}, these are 2-simplices in $\Delta^n$ and so
  every horn has a unique extension as required.
\end{proof}

We note that \cref{cor:ttt:delta-is-category} is already a significant step forward for \STT{}: it
is the first result constructing an explicit example of a non-discrete category within the system.

\subsection{Reasoning with modalities in \texorpdfstring{\TTT{}}{Triangulated Type Theory}}

A number of useful results in \TTT{} are immediate corollaries of standard results from \MTT{}
combined with one of the axioms. We record some of the most important results in this section for
future use and to give a flavor for how modalities can be used to enhance simplicial reasoning.

By general results about adjoint modalities from \MTT{}~\parencite{gratzer:mtt-journal:2021}, we obtain
the following:
\begin{lemma}
  \label{lem:ttt:transpose}
  \leavevmode
  \begin{itemize}
  \item If $\DeclVar{A}{\ArrId{}}{\Uni}$, $\DeclVar{B}{\SM}{\Uni}$ there is an equivalence
    $\Modify[\SM]{\Modify[\GM]{A} \to B} \Equiv \prn{A \to \Modify[\SM]{B}}$.
  \item If $\DeclVar{A}{\ArrId{}}{\Uni}$, $\DeclVar{B}{\OM}{\Uni}$ there is an equivalence
    $\Modify[\OM]{\Modify[\OM]{A} \to B} \Equiv \prn{A \to \Modify[\OM]{B}}$.
 \end{itemize}
  There are also dependent versions where \eg{},
  $\DeclVar{B}{\SM}{\MFn[\GM]{A}{\Uni}}$.
\end{lemma}

We record two useful consequences of the transposition principle for $\OM \Adjoint \OM$ and $\GM \Adjoint \SM$:
\begin{lemma}[\textcite{gratzer:phd}]
  \label{lem:ttt:o-cocont}
  $\Modify[\OM]{-}$ commutes with colimits.
\end{lemma}
\begin{lemma}
  \label{lem:ttt:sharp-codiscrete}
  Evaluation at endpoints $\prn{\Int \to \Modify[\SM]{A}} \to \prn{\Bool \to \Modify[\SM]{A}}$ is an
  equivalence.
\end{lemma}

A similar result is available for $\Int \to -$ in light of \cref{ax:int-is-left-adjoint}:
\begin{lemma}
  \label{lem:ttt:amazing-transpose}
  There is a unique map $\DeclVar{{-}_\Int}{\GM}{\MFn[\GM]{\Uni}{\Uni}}$ such that the
  following bijection holds:
  \[
    \Prod{\DeclVar{A,B}{\GM}{\Uni}} \Modify[\GM]{A^\Int \to B} \Equiv \Modify[\GM]{A \to B_\Int}
  \]
\end{lemma}

Consequently, ${-}_\Int$ preserves limits. Notably, ${\ObjTerm{}}_\Int = \ObjTerm{}$ and ${-}_\Int$
commutes with taking fibers. These two facts imply that $\prn{\Sum{A : \Uni} A}_\Int \to \Uni_\Int$
has small fibers and induces a dependent version of this operation ${-}_\Int : \Uni_\Int \to \Uni$.

In more detail, there is a canonical functorial action of $\prn{-}_\Int$ which induces a map
$\prn{\Sum{A : \Uni[i]} A}_\Int \to \prn{\Uni[i]}_\Int$. Viewed as a family of types, we therefore obtain a
map $\prn{\Uni[i]}_\Int \to \Uni[i + 1]$. We can argue that this map actually factors through
$\Uni[i]$ rather than just $\Uni[i + 1]$ as follows. First, by univalence the map $\Uni[i] \to
\Uni[i + 1]$ is an embedding, so a factorization of $\prn{\Uni[i]}_\Int \to \Uni[i + 1]$ is unique
when it exists. Using \cref{ax:cubes-separate}, it suffices to check that such a factorization
exists after restricting $\prn{\Uni[i]}_\Int \to \Uni[i + 1]$ along some map
$\DeclVar{A}{\GM}{\Int^n \to \prn{\Uni[i]}_\Int}$. Using the universal property of $\prn{-}_\Int$, we
may factor $A$ through a map $\tilde{A}_\Int$ where
$\DeclVar{\tilde{A}}{\GM}{\prn{\Int^n}_\Int \to \prn{\Uni[i]}_\Int}$. Since $\prn{-}_\Int$ preserves
fibers, it suffices to show that the pullback $P$ in the following is $\Uni[i]$-small:
\[
  \begin{tikzpicture}[diagram]
    \node[pullback] (P) {$P$};
    \node[pullback, right = 3cm of P] (XInt) {$\prn{\Sum{f : \prn{\Int^n}^\Int} \tilde{A}\prn{f}}_\Int$};
    \node[below = of P] (Intn) {$\Int^n$};
    \node[right = 3cm of Intn] (IntnInt) {$\prn{\prn{\Int^n}^\Int}_\Int$};
    \node[right = 3cm of IntnInt] (U) {$\prn{\Uni[i]}_\Int$};
    \node[right = 3cm of XInt] (UdotInt) {$\prn{\Sum{A : \Uni[i]} A}_\Int$};
    \path[->] (P) edge (XInt);
    \path[->] (P) edge (Intn);
    \path[->] (Intn) edge (IntnInt);
    \path[->] (IntnInt) edge (U);
    \path[->] (XInt) edge (UdotInt);
    \path[->] (UdotInt) edge (U);
    \path[->] (XInt) edge (IntnInt);
  \end{tikzpicture}
\]
Finally, we note that since the ordinary version of $\prn{-}_\Int$ preserves $\Uni[i]$ small
types, $P$ is the pullback of $\Uni[i]$-small types and is therefore itself $\Uni[i]$-small.

To give an example of how these reasoning principles can be used, we show how they can be used to
enhance our stock of simplicial types.

\begin{lemma}
  Given $\DeclVar{A}{\OM}{\Uni}$, if $\Modify[\OM]{\IsSimp\prn{A}}$ then $\IsSimp\prn{\Modify[\OM]{A}}$.
\end{lemma}
\begin{proof}
  Fix $i,j : \Int$ such that we must show
  $\Modify[\OM]{A} \to \prn{i \le j \lor j \le i \to \Modify[\OM]{A}}$ is an equivalence. Using
  \cref{lem:ttt:transpose}, the codomain is equivalent to $\Modify[\OM]{\Modify[\OM]{i \le j \lor j \le i} \to A}$.
  By \cref{ax:op-of-int-is-int,lem:ttt:o-cocont}, $\Modify[\OM]{i \le j \lor j \le i}$ is
  $\neg \MkMod[\OM]{i} \ge \neg \MkMod[\OM]{j} \lor \neg \MkMod[\OM]{j} \ge \neg \MkMod[\OM]{i}$ and the
  conclusion follows immediately from our assumption $\Modify[\OM]{\IsSimp\prn{A}}$.
\end{proof}

\begin{lemma}
  \label{thm:ttt:global-disc-is-simplicial}
  If $\DeclVar{A}{\GM}{\Uni}$ is discrete then $A$ is simplicial.
\end{lemma}
\begin{proof}
  Assume $A$ is discrete, \ie, $A \to A^\Int$ is an equivalence.  Since cubes
  separate by~\cref{ax:cubes-separate}, it suffices to show for all polynomials
  $p,q : \Int\brk{\vec x}$ in $n$ variables $\vec x$ that the map $A \to (\varphi(\vec x) \to A)$ is
  an equivalence, where $\varphi(\vec x) := p(\vec x)\le q(\vec x) \lor q(\vec x)\le p(\vec x)$.

  In turn, it suffices to give an $\Int$-homotopy $h$ connecting the constant map at $0$ to the
  identity on $\varphi(\vec x)$, for each $\vec x:\Int^n$. We notice that the straight-line homotopy
  $h(\vec x,t) = \vec x\land t$ from $0$ to $\vec x$ works: We have to show for each $\vec x$ with
  $\varphi(\vec x)$ that $\varphi(\vec x\land t)$ holds, for each $t$.  But notice that $\varphi(0)$
  is true, as any pair of constants among $0,1$ are comparable.  By \cref{lem:ttt:phoa},
  $\varphi(\vec x\land t)$ then holds for all $t$.
\end{proof}
\noindent
Using the adjunction $\GM \Adjoint \SM$, we can prove that \eg{},
$\Nat \Equiv \Modify[\GM]{\Nat}$~\parencite{gratzer:phd}. Accordingly, by \cref{ax:discrete-iff-crisp}:
\begin{corollary}
  $\Nat$ and $\Bool$ are both simplicial and $\Int$-null \ie{} groupoids.
\end{corollary}

\begin{remark}
  The result analogous to~\Cref{thm:ttt:global-disc-is-simplicial} for Rezk-complete Segal types
  does not hold, falsifying a conjecture of \textcite{weaver:2020}. In particular, $\Delta^2
  \Pushout{\Int} \Delta^2$ can be shown to be Rezk-complete and Segal, but is not simplicial. This
  same example shows that the requirement that $A$ be annotated with $\GM$ is necessary: as a family
  over $\Int \times \Int$ the type $\Delta^2 \Pushout{\Int} \Delta^2$ is fiberwise a
  proposition---explicitly, it is $\lambda i\,j.\,i \le j \lor j \le i$---and therefore it is
  fiberwise $\Int$-null. If we could apply \cref{thm:ttt:global-disc-is-simplicial} without the
  $\GM$-annotation we could conclude that each fiber $i \le j \lor j \le i$ was simplicial. Combined
  with the fact that $\Int \times \Int$ is simplicial, this leads again to the false conclusion that
  $\Delta^2 \Pushout{\Int} \Delta^2$ is simplicial.
\end{remark}

\section{The cubical spaces model}
\label{sec:model}

\TTT{} is intended to be an internal language for cubical sets \ie{} $\PSH{\CUBE}$ (or rather its
$\infty$-categorical enhancement). In order to make this precise, we construct a model of \TTT{} in
which types are realized as (families of) $\infty$-presheaves over the Dedekind cube category.
Immediately, we must contend with the fact that syntax is 1-categorical and models of
syntax~\parencite{gratzer:mtt-journal:2021} are also inherently 1-categorical. To overcome this
mismatch, we interpret \TTT{} into a model category which presents the appropriate presheaf
$\infty$-category~\parencite{shulman:elegant:2015,kapulkin:2021,shulman:2019}.

For \TTT{}, this model category will be the \emph{injective model structure on simplicial presheaves
on Dedekind cubes} $\PSH[\SSET]{\CUBE}$. That is, types in \TTT{} are interpreted as certain
families of presheaves over $\CUBE$ valued in $\SSET = \PSH{\SIMP}$. It is helpful to view the
simplicial sets layer as ``mixing in'' homotopy theory with ordinary presheaves over $\CUBE$.

The construction of the intended model of \TTT{} is largely an exercise in combining off-the-shelf
results about models of \HOTT{} and models of \MTT{}. In particular, \textcite{shulman:2019} shows
that \HOTT{} admits a model in $\PSH[\SSET]{\CUBE}$ and results of \textcite{shulman:2023} and
\textcite{gratzer:phd} show that this model refines to a model of \MTT{}. One then directly verifies
that this model validates the additional axioms required by \TTT{}. However, for the sake of
completeness we give some of the details of this process. In particular, in \cref{sec:model:def} we
spell out the definition of a model of \MTT{}. In \cref{sec:model:constr} we show how to apply a
theorem of \textcite{shulman:2023} to construct a model of \MTT{} and verify that it extends
appropriately to the desired model of \TTT{}.

\begin{remark}
  Models of \TTT{} specifically and type theory generally contain quite a lot of data to account for
  each connective. We will focus primarily on the specific modal connectives in \TTT{} in this
  section, as---just as in the syntax of \TTT{} itself---the other connectives are treated in a
  totally standard manner.
\end{remark}

\subsection{Models of \texorpdfstring{\MTT{}}{MTT} and \texorpdfstring{\TTT{}}{TTT}}
\label{sec:model:def}

In this subsection, we briefly recall the main aspects of the model theory of \MTT{}. We shall not
directly use these definitions in our construction of the cubical spaces model of \TTT{} and instead
will rely on a general \emph{coherence} result of Shulman. Accordingly, these
definitions are recalled only to make the following discussion more concrete.

When originally introduced by \textcite{gratzer:mtt-journal:2021}, \MTT{} was presented as a certain
\emph{generalized algebraic theory} (GAT). Consequently, the general theory of GATs ensures that
there is a category of models of \MTT{} (parameterized by the chosen mode theory) which refines the
corresponding category of models of ordinary Martin-L{\"o}f type theory~\parencite{dybjer:1996}.
Let us therefore begin by recalling the definition of a model of dependent type theory:
\begin{definition}[\textcite{awodey:2018}]
  A category with families (CwF) $\prn{\CC,\tau}$ consists of a category $\CC$ along with a morphism
  $\Mor[\tau]{U^\bullet}{U}$ in $\PSH{\CC}$ is equipped with the following chosen
  data:\footnote{The fact that this is chosen data rather than a mere existence property is a quirk
    of generalized algebraic theories. More refined recent approaches such as those given by
    \textcite{uemura:phd} do not have this deficiency.}
  \begin{itemize}
    \item $\CC$ has a chosen terminal object,
    \item $\tau$ is locally representable \ie{}, the fiber $\Yo{c} \times_{U} U^\bullet$ has a
      chosen representation $\Yo{c'}$ for each morphism $\Mor{\Yo{c}}{U}$.
  \end{itemize}
\end{definition}

This structure is a model of dependent type theory with no connectives, only basic operations like
context extension, variables, and substitutions. We refer the reader to \textcite{awodey:2018} for a
careful exposition of how connectives may be integrated into this definition. Since this aspect of
CwFs is carried through unchanged through the remainder of our discussion, we ignore it here.

A model of \MTT{} elaborates on this structure by linking together CwFs via functors (which
intuitively model $\Gamma/\mu$). We begin with the structure necessary to model \MTT{} without any
connectives.
\begin{definition}
  A model of \MTT{} without any connectives consists of a strict 2-functor
  $\Mor[F]{\Coop{\Mode}}{\CAT}$ along with a choice of morphism $\Mor[\tau_m]{U^\bullet_m}{U_m}$ in
  $\PSH{F\prn{m}}$ for each $m : \Mode$. We require the following additional data:
  \begin{itemize}
    \item A chosen terminal object for each $F\prn{m}$,
    \item A choice of local representability structure on $F\prn{\mu}^*\prn{\tau_m}$ for each
      $\Mor[\mu]{m}{n}$.
  \end{itemize}
\end{definition}

Note that $F\prn{\ArrId{}}^*\prn{\tau_m} = \tau_m$ so each $\prn{F\prn{m},\tau_m}$ is a CwF. The
additional requirement that each $F\prn{\mu}^*\prn{\tau_m}$ be locally representable is used to model
\emph{annotated} variables in a context \ie{} $\Gamma, \DeclVar{x}{\mu}{A}$.
Other connectives are integrated into this definition without change. For instance, the inclusion of
$\Sum{}$-types is accounted for by requiring that each $\tau_m$ is closed under $\Sum{}$-types. This
process is entirely mechanical and their interpretation in the intended model of \TTT{} is
unsurprising, accordingly we refer the reader to \textcite{gratzer:mtt-journal:2021} for further
information.

The central novel connective of \MTT{}---the modal types---is slightly more complex, and so we
briefly touch on their definition. We refer the reader to \textcite{gratzer:mtt-journal:2021} or
\textcite{gratzer:phd} for a careful discussion of how the following definition relates to the
syntax introduced in \cref{sec:ttt}.

\begin{definition}
  A model of \MTT{} without any connectives
  $\prn{\Mor[F]{\Coop{\Mode}}{\CAT}, \prn{\tau_m}_{m : \Mode}}$ supports modal types when equipped
  with the following data for each $\Mor[\mu]{n}{m}$:
  \begin{itemize}
    \item A commuting square $\Mor[\alpha]{F\prn{\mu}^*\prn{\tau_n}}{\tau_m}$ in $\PSH{F\prn{m}}$.
    \item Writing
      $\Mor[m]{F\prn{\mu}^*{\prn{U^\bullet_n}}}{F\prn{\mu}^*{\prn{U^\bullet_n}} \times_{U_m} U^\bullet_m}$
      for the gap-map over $F\prn{\mu}^*U_n$ induced by $\alpha$, we require a \emph{stable weak orthogonality
      structure}~\parencite{awodey:2018} $s : m \Orth \prn{F\prn{\mu}^*U_n}^*\prn{\tau_m}$ in
      $\SLICE{\PSH{F\prn{m}}}{F\prn{\mu}^*U_n}$.
  \end{itemize}
  Roughly, the first point encodes the introduction and formation rules of the $\mu$ modal type and
  the second encodes the elimination rule and its attendant equality.
\end{definition}

We will not spend much time with this definition because, by a result of \textcite{shulman:2023},
one can construct a model of \MTT{} with all connectives from far more recognizable data. In
particular, \opcit{} adapts the local universes coherence construction of \textcite{lumsdaine:2015}
which promotes type-theoretic fibration categories to models.

\begin{definition}
  A type-theoretic fibration category consists of a category $\mathcal{C}$ with all finite limits
  with a chosen class of morphisms $\mathcal{F}$ referred to as fibrations which satisfy the
  following closure conditions:
  \begin{itemize}
    \item $\mathcal{F}$ is closed under identity and composition.
    \item $\mathcal{F}$ is closed under pullback along arbitrary maps.
    \item $\mathcal{F}$ is closed under pushforwards; if $f \in \mathcal{F}$ then the right adjoint
      to pullback along $f$ sends fibrations to fibrations.
    \item Every morphism $f$ admits a factorization $f' \circ i$ such that $f'$ is a fibration and
      $i$ is weakly left orthogonal to fibrations and this factorization is stable under pullbacks.
  \end{itemize}
  We refer to maps weakly left orthogonal to fibrations as \emph{anodyne}.
\end{definition}

\begin{definition}
  A type-theoretic fibration category has a universe if there is a map $\Mor[\tau]{U^\bullet}{U}$
  such that the following conditions hold:
  \begin{itemize}
    \item Both $\tau$ and $\Mor{U}{\ObjTerm{}}$ are fibrations.
    \item Fibrations arising from pulling back $\tau$ ($U$-small fibrations) satisfy all but the
      last closure condition for fibrations.
    \item If $\Mor[f]{X}{Y}$ is a $U$-small fibration, there is a stable factorization of
      $\Mor[\Delta_f]{X}{X \times_Y X}$ as $\Delta'_f \circ i$ where $i$ is weakly orthogonal to
      fibrations and $\Delta'_f$ is a $U$-small fibration.
  \end{itemize}
  A hierarchy of universes is given by a collection of universes $\Mor[\tau_i]{U^\bullet_i}{U_i}$
  such that each $\tau_i$ and $\Mor{U_i}{\ObjTerm{}}$ are $U_{i + 1}$-fibrations.
\end{definition}

\begin{proposition}[\textcite{shulman:2023}, \textcite{gratzer:phd}]
  \label{thm:model:coherence}
  A model of \MTT{} with mode theory
  $\mathcal{M}$ can be constructed from the following pieces of data:
  \begin{itemize}
    \item A pseudofunctor $\Mor[F]{\mathcal{M}}{\CAT}$ such that each $F\prn{\mu}$ has a right
      adjoint $G_\mu$,

    \item For each $m : \Mode$, a choice of arrows making $F\prn{m}$ into a type-theoretic fibration
      category with a hierarchy of universes.

    \item For each $\Mor[\mu]{n}{m}$ and fibration $\Mor[f]{X}{Y}$ in $F\prn{n}$, there is a chosen
      stable factorization $G_\mu\prn{f} = m\brk{\mu}_f \circ i\brk{\mu}_f$ such that
      $\Mor[m\brk{\mu}_f]{X'}{G_\mu\prn{Y}}$ is a fibration and
      $\Mor[i\brk{\mu}_f]{G_\mu\prn{X}}{X'}$ is anodyne.
      Additionally, we require that if  $\Mor[\nu]{m}{o}$ and $\Mor[g]{Z}{G_\nu\prn{G_\mu\prn{Y}}}$,
      the map
      $\Mor[g^*\prn{G_\nu\prn{i\brk{\mu}_f}}]{%
        Z \times_{G_\nu\prn{G_\mu\prn{Y}}} G_\nu\prn{G_\mu\prn{X}}
      }{%
        Z \times_{G_\nu\prn{G_\mu\prn{Y}}} G_\nu\prn{X'}
      }$ is anodyne.
  \end{itemize}
\end{proposition}

In this model, contexts and types at mode $m$ are modeled by objects and fibrations of $F\prn{m}$.
Consequently, our goal is to apply \cref{thm:model:coherence} to the category of cubical spaces and
to show that the resulting model of \MTT{} satisfies the axioms required of \TTT{}.

\subsection{Constructing the cubical spaces model}
\label{sec:model:constr}

First, we require the following result fundamental result from \textcite{shulman:2019}:

\begin{definition}
  Let $\CUBE$ be the category of Dedekind cubes \ie{}, the full subcategory of $\CAT$ generated by
  finite products of the category $\brk{1} = \brc{0 \to 1}$.
\end{definition}

\begin{proposition}[\textcite{shulman:2019}]
  $\PSH[\SSET]{\CUBE}$ and the collection of injective fibrations forms a type-theoretic fibration
  category with a hierarchy of univalent universes.
\end{proposition}

We will now use this type-theoretic fibration category as the basic input for
\cref{thm:model:coherence}. In particular, we consider the functor $\Mor[F]{\TTTMode}{\CAT}$ which
sends the unique object $m$ to $\PSH[\SSET]{\CUBE}$ and interprets the 1-cells as follows:
\[
  F\prn{\GM}\prn{X} = \brk{n} \mapsto X\prn{\brk{0}}
  \qquad
  F\prn{\OM}\prn{X} = \brk{n} \mapsto X\prn{\Op{\brk{n}}}
  \qquad
  F\prn{\SM}\prn{X} = \brk{n} \mapsto X\prn{\brk{0}}^n
\]
Here we have written $\Op{-}$ for the unique functor $\Mor{\CUBE}{\CUBE}$ which sends $\brk{n}$ to
$\brk{n}$ but exchanges $0$ for $1$. We may directly check that all of the required equalities are
satisfied and it is also clear that $F\prn{\GM} \Adjoint F\prn{\SM}$. It remains to define $F$ on
the generating inequalities $\GM \le \ArrId{}$ and $\ArrId{} \le \SM$. We realize these by the
unique morphism in $\Hom[\CUBE]{\brk{n}}{\brk{0}}$ (the counit of the aforementioned
adjunction) and its transpose (the unit).

In order to show that this functor is well-defined, we must show that it sends all 2-cells
$\Mor[\alpha]{\mu}{\nu}$ to the same natural transformations
$\Mor[F\prn{\alpha}]{F\prn{\mu}}{F\prn{\nu}}$. For instance, we must argue that the two 2-cells
$\Mor{\GM \circ \GM}{\GM}$ induced by $\GM \le \ArrId{}$ and whiskering on either side are sent to
the same natural transformation. A priori, this is far from obvious: the presence of various
equalities between modalities in $\TTTMode$ allows for non-obvious 2-cells. In this particular
model, however, our task is far easier. The fact that $\CUBE$, $\SSET$, and $\SET$ all have trivial
centers ensures that $\PSH[\SSET]{\CUBE}$ has a trivial center and we prove the following:
\begin{lemma}
  For each $\Mor[\mu,\nu]{m}{m}$ in $\TTTMode$, there is at most one natural transformation
  $\Mor{F\prn{\mu}}{F\prn{\nu}}$.
\end{lemma}
\begin{proof}
  First, we note that the equations for $\TTTMode$ ensure that we need only consider
  $\mu,\nu \in \brc{\OM,\SM,\GM,\ArrId{}}$. Further, by adjointness, we may ignore the cases where
  $\nu \in \brc{\SM,\OM}$. We check the remaining cases directly.

  For instance, if $\mu = \nu = \ArrId{}$, then we note that a natural transformation
  $\Mor{F\prn{\mu}}{F\prn{\nu}}$ is determined by its behavior at $\Yo{\brk{1}}$; both sides
  preserve products, tensoring by simplicial sets, and colimits and all cubical spaces are generated
  from $\Yo{\brk{1}}$ under these operations. However, there is only one morphism
  $\Mor{\brk{1}}{\brk{1}}$ which preserves endpoints: the identity map. The same reasoning rules out
  the possibility of any natural transformations with $\mu = \ArrId{},\nu = \OM$ or $\mu = \ArrId{},
  \nu = \GM$. It also guarantees that there is exactly one natural transformation
  $\Mor{F\prn{\ArrId{}}}{F\prn{\SM}}$.

  The case where $\mu = \GM$ is similar, so we focus on the remaining case where $\mu = \SM$. This
  case is slightly more complex, since $F\prn{\SM}$ does not preserve colimits and so we can
  immediately reduce to checking behavior at $\Yo{\brk{1}}$. However, we can immediately rule out
  natural transformations $\Mor{F\prn{\SM}}{F\prn{\ArrId{}}}$, $\Mor{F\prn{\SM}}{F\prn{\OM}}$, or
  $\Mor{F\prn{\SM}}{F\prn{\GM}}$ by considering the behavior of such morphisms at $\Yo{\brk{1}}$.
  To show that there is exactly one morphism $\Mor{F\prn{\SM}}{F\prn{\SM}}$, we note that it
  suffices to show that there is exactly one natural transformation
  $\Mor[\alpha]{\brk{0}^*}{\brk{0}^*}$ as functors from $\Mor{\PSH[\SSET]{\CUBE}}{\SSET}$. Now these
  functors again preserve limits, colimits, and tensoring by simplicial sets and so we reduce to
  considering the behavior of $\alpha$ on $\brk{0}$ and $\brk{1}$, where we see it must once more be
  the identity.
\end{proof}

With this well-defined functor $\Mor[F]{\TTTMode}{\CAT}$, we immediately check that each 1-cell in
$\TTTMode$ is sent to a right adjoint and so there is a conjugate functor
$\Mor[\bar{F}]{\Coop{\TTTMode}}{\CAT}$ and it is to this we apply \cref{thm:model:coherence}. The
first requirement of this theorem (that each 1-cell is sent by $\bar{F}$ to a right adjoint) is
automatic in this case. The previously cited result of Shulman already shows that $F\prn{m}$ has an
appropriate choice of type-theoretic model structure. It remains only to handle the third point,
which governs the interpretation of modalities into this putative model. For three of the four
modalities we must interpret, this is trivial:

\begin{lemma}
  For $\mu \in \brc{\SM,\OM,\ArrId{}}$, $F\prn{\mu}$ is a right Quillen functor for both the
  injective and projective model structures.
\end{lemma}

In particular, given a fibration $\Mor[f]{X}{Y}$ in $\bar{F}\prn{m}$, we may trivially factorize \eg{},
$\Mor{F\prn{\SM}\prn{X}}{F\prn{\SM}\prn{Y}}$ into an anodyne map (trivial injective cofibration)
followed by an (injective) fibration, simply by taking the identity map followed by
$F\prn{\SM}\prn{f}$. This factorization satisfies all the desiderata of the third point of
\cref{thm:model:coherence} trivially. Unfortunately, $F\prn{\GM}$ is not a right Quillen functor
for the injective model structure and a more elaborate approach required in this case.

\begin{lemma}
  Given an injective fibration $\Mor[p]{X}{Y}$, there is a factorization of $F\prn{\GM}\prn{f}$ into
  a trivial cofibration $i$ followed by an injective fibration $f$ such that (1) this factorization
  is stable under pullback and (2) if $\Mor[\nu]{m}{m}$ in $\TTTMode$ and
  $\Mor[z]{Z}{F\prn{\nu}\prn{F\prn{\mu}\prn{Y}}}$ then $z^*\prn{F\prn{\nu}\prn{i}}$ is a trivial
  cofibration.
\end{lemma}

\begin{proof}
  We begin by noting that while $p' = F\prn{\GM}\prn{p}$ is not necessarily an injective fibration,
  it is a \emph{projective} fibration. Indeed, $p$ is a projective fibration so that, by definition,
  $p\prn{\brk{0}}$ is a fibration and so $p'$ is a levelwise (\ie{}, projective) fibration. We may
  now apply the cobar construction detailed by \textcite[Definition 8.17]{shulman:2019} to obtain a
  stable factorization of $p'$ into a trivial cofibration $i$ followed by an injective fibration
  $f$, as required.

  For the second condition, we note that $p'$ and $f$ are both projective fibrations and,
  consequently, so too are $F\prn{\nu}\prn{p'}$ and $F\prn{\nu}\prn{f}$. Moreover, since 
  injective trivial cofibrations are precisely a levelwise trivial cofibration, we may also conclude
  that $F\prn{\nu}\prn{i}$ is a injective trivial cofibration---note that trivial cofibrations are
  closed under cartesian products in $\SSET$ for the case where $\nu = \SM$. Finally, since
  pullbacks are computed levelwise and the model structure on $\SSET$ is right proper, we see that
  $z^*\prn{F\prn{\nu}\prn{i}}$ remains a trivial cofibration for all
  $\Mor[z]{Z}{F\prn{\nu}\prn{F\prn{\mu}\prn{Y}}}$.
\end{proof}

\begin{lemma}
  There is a model of \MTT{} with mode theory $\TTTMode$ in $\PSH[\SSET]{\CUBE}$ which
  interprets modalities using the functors described by $F$
\end{lemma}

Finally, we must show that this model of \MTT{} validates the axioms necessary for \TTT{}.
\cref{ax:univalence} is an immediate consequence of \textcite{shulman:2019}. Notably,
\textcite{gratzer:phd} shows that \cref{ax:crisp-id-induction} holds for $\SM$, $\OM$, and
$\ArrId{}$---since they are \emph{dependent right adjoints}---and that it holds for $\GM$ since it
is an internal left adjoint.

\cref{ax:int,ax:op-of-int-is-int,ax:global-points,ax:sqc} are all statements about
sets---$0$-truncated types---and so hold in the above model if and only if they hold in the ordinary
interpretation of type theory in $\PSH{\CUBE}$. That is, for these axioms we may ignore
the simplicial dimension of $\PSH[\SSET]{\CUBE}$. In this case, each of these except
\cref{ax:sqc} is a routine verification.\footnote{See also \parencite{myers:2025} for a general topos-theoretic account.} The duality axiom, finally, follows from a result of
\textcite[Theorem 4.11]{blechschmidt:2023} combined with the following lemma:
\begin{lemma}
  There is a canonical geometric embedding $\Mor{\PSH{\CUBE}}{\SET\brk{\mathrm{DLat}}}$ of cubical
  sets into the classifying topos of distributive lattices. Under this embedding, $\Yo{\brk{1}}$
  in $\PSH{\CUBE}$ is sent to the generic bounded distributive lattice.
\end{lemma}

The remaining axioms
(\cref{ax:int-is-left-adjoint,ax:discrete-iff-crisp,ax:cubes-separate,ax:simplicial-stability}) are
verified by straightforward (if tedious) categorical arguments. In each case, we must show
that a certain function in type theory is an equivalence. After unfolding to the model, this amounts
to showing that a certain morphism is a weak equivalence. In each case, these statements admit
natural $\infty$-categorical interpretations (\eg{}, \cref{ax:int-is-left-adjoint} postulates the
existence of an amazing right adjoint to $\prn{-}^\Int$), but a model-categorical argument is
necessary to connect this fact to the interpretation of \TTT{} in $\PSH[\SSET]{\CUBE}$. Since these
arguments are nearly all the same, we choose to focus on the most complex:
\cref{ax:int-is-left-adjoint}.

\begin{lemma}
  The following type is inhabited in the model of \MTT{} in $\PSH[\SSET]{\CUBE}$:
  \[
    \prn{\DeclVar{A}{\GM}{\Uni}} \to
    \Sum{\DeclVar{A_\Int}{\GM}{\Uni}} \Sum{\DeclVar{\epsilon}{\GM}{\prn{A_\Int}^\Int \to A}}
    \Prod{\DeclVar{B}{\GM}{\Uni}} \IsEquiv\prn{\Modify[\GM]{B \to A_\Int} \to \Modify[\GM]{B^\Int \to A}}
  \]
\end{lemma}
\begin{proof}
  For concision, let us denote the codomain of the above function by $\Con{isRepr}\prn{A}$.
  First, we note that we may prove directly in \MTT{} that $\Con{isRepr}\prn{A}$ is a proposition,
  so this type is inhabited if and only if the following fibration is a trivial fibration in our
  model:
  \[
    \Mor{\Interp{\Sum{\MkMod[\GM]{A} : \Modify[\GM]{\Uni}} \Con{isRepr}\prn{A}}}{\Interp{\Modify[\GM]{\Uni}}}
  \]

  To prove this, we may restrict our attention to showing that the restriction of this fibration
  along the trivial cofibration
  $\Mor{\Con{Const}\prn{\Interp{\Uni}\prn{\brk{0}}}}{\Interp{\Modify[\GM]{\Uni}}}$ is trivial. Here $\Con{Const}\prn{X}$ denotes the constant cubical space with value $X$.

  It suffices to check that this morphism is a trivial fibration fiberwise. We therefore need to
  show that the restriction of this map along
  $\Mor[A]{\Yo{\brk{n}}}{\Con{Const}\prn{\Interp{\Uni}\prn{\brk{0}}}}$ is a trivial fibration. However, 
  $A$ factors through $\Yo{\brk{0}}$. Consequently, we may restrict our attention to the case where
  $n = 0$.

  Accordingly, we may fix $X \in \Interp{\Uni}\prn{\brk{0}}$ \ie{} an injectively fibrant presheaf
  $X$ and we must show that $\Interp{\Con{isRepr}\prn{A}}_{A \mapsto X}$ is inhabited. We now unfold
  $\Con{isRepr}$ to see that it suffices to construct another injectively fibrant presheaf $Y$ such
  that for all fibrant $Z$, there is a natural weak equivalence
  $\Mor{\Hom{Z^{\Yo{\brk{1}}}}{X}}{\Hom{Z}{X'}}$.

  Let us now note that $\prn{-}^{\Yo{\brk{1}}} \cong \Pre{\prn{- \times \brk{1}}}$. This is a left adjoint
  which preserves injective cofibrations and therefore has a right adjoint $\prn{-}_{\Yo{\brk{1}}}$
  which preserves injective fibrations. We choose $X' = \prn{X}_{\Yo{\brk{1}}}$ and the rest
  follows.
\end{proof}

\begin{theorem}
  \TTT{} has a model in $\PSH[\SSET]{\CUBE}$ where types are injective fibrations and modalities are
  interpreted as described above.
\end{theorem}
Crucially, within this model simplicial types are precisely those belonging to the subtopos
$\PSH[\SSET]{\SIMP}$~\parencite{streicher:2021}. Consequently, the adequacy result from
\textcite{riehl:2017} applies and we conclude that this model shows that any fact proven about
categories and groupoids inside of \TTT{} is a valid proof for the standard definition of
$\infty$-categories.
\begin{theorem}
  Categories in \TTT{} adequately model $\infty$-categories.
\end{theorem}

\section{Covariant and amazingly covariant families}
\label{sec:covariance}

In \cref{sec:stt}, we saw how groupoids were defined internally as those types satisfying
$\Con{isGroupoid}\,A = \IsEquiv\prn{A \to A^\Int}$. We might hope this induces a
directed univalent universe of groupoids directly, by considering
$\Uni_{\Con{grp}} = \Sum{A : \Uni} \Con{isGroupoid}\,A$. However, this is far from our desired
universe. Most glaringly, while $F : A \to \Uni_{\Con{grp}}$ is a family of groupoids over $A$, this
family is not required to respect the category structure of $A$ in any way. In fact, one may show
that a map $F : \Int \to \Uni_{\Con{grp}}$ is akin to an unstructured relation between $F\prn{0}$
and $F\prn{1}$ and nothing like the function required for directed univalence.

\begin{example}
  By assumption, $\Int$ is a set and so $f = \lambda i.\,i = 0$ is a function $\Int \to \Prop$. Since
  each proposition is a groupoid, this ensures that $f$ factors through $\Uni_{\Con{grp}}$ despite
  the fact that there can be no function from $f\prn{0} \to f\prn{1}$.
\end{example}

In order to rectify this and define $\Space$, we shall require a theory of families of groupoids
where a morphism $f : \Hom{a}{a'}$ in $A$ induces a functor of groupoids $F\prn{a} \to F\prn{a'}$.
\textcite{riehl:2017} termed these \emph{covariant} families and they are further studied by
\textcite{buchholtz:2023}. As mentioned in the introduction, we shall also require a modal version
of covariant families $F : A \to \Uni$ which are covariant not only in $A$ but also in the entire
context.

\subsection{Covariant families and transport}
\label{sec:covariance:covariant}

We begin by recalling the definition of a covariant family from \textcite{riehl:2017}.

\begin{definition}
  A family $A : X \to \Uni$ is \emph{covariant} if the following proposition holds:
  \[
    \IsCovFam\prn{A} =
    \Prod{x : \Int \to X}
    \Prod{a_0 : A\prn{x\,0}}
    \IsContr\prn{
      \Sum{a_1 : A\prn{x\,1}} \Hom[A]<x>{a_0}{a_1}
    }
  \]
\end{definition}

\begin{convention}
  While not strictly necessary, we will assume that the base of a covariant family $A$ is a Segal
  type unless explicitly noted otherwise.
\end{convention}

We recall a few facts about covariant families (also due to \textcite{riehl:2017}).

\begin{lemma}
  If $A : X \to \Uni$ is covariant and $f : Y \to X$ then $A \circ f$ is also covariant.
\end{lemma}
\begin{lemma}\label{lem:covariance:transport}
  Given $\phi : \IsCovFam\prn{A : X \to \Uni}$ and $f : \Hom{x_0}{x_1}$ then there is an induced
  transport map $\Coe_{A \circ f} : A\prn{x_0} \to A\prn{x_1}$. Moreover, transport maps respect
  composition and identities.
\end{lemma}
\begin{proof}[Proof Sketch]
  One defines the transport map $\Coe_{A \circ f}\prn{a_0} = \Proj[1]{\prn{\phi\,f\,a_0}}$.
  We leave it to the reader to check that this has the appropriate type and that the expected
  identities are satisfied.
\end{proof}

\begin{lemma}
  \label{lem:covariance:covariant-implies-groupoid} Given $\phi : \IsCovFam\prn{A : X \to \Uni}$ and
  $x : X$, the fiber $A\prn{x}$ is a groupoid.
\end{lemma}
\begin{proof}
  Since covariant families stable under base-change, $A\prn{x}$ is a covariant family $\ObjTerm{}
  \to \Uni$. Unfolding definitions, we conclude that the following holds:
  \[
    \prn{a : A\prn{x}} \to \IsContr\prn{\Sum{f : \Int \to A\prn{x}} f\prn{0} = a}
  \]
  This is equivalent to the proposition that
  $\Con{eval}\prn{0} : \prn{\Int \to A\prn{x}} \to A\prn{x}$ is an equivalence. By 3-for-2, this
  implies that $\Con{const} : A\prn{x} \to A\prn{x}^\Int$ is an equivalence, as
  $\Con{eval}\prn{0} \circ \Con{const} = \ArrId{}$.
\end{proof}

It is often helpful to rephrase covariant families in terms of orthogonality conditions:
\begin{definition}
  Given a type family $A : X \to \Uni$, we shall write $\tilde{A}$ for the \emph{total type}
  $\Sum{x : X} A\prn{x}$.
\end{definition}
\begin{lemma}[\protect{\cite[Theorem 8.5]{riehl:2017}}]
  A family $A : X \to \Uni$ is covariant if and only if the projection map $\tilde{A} \to X$ is
  right orthogonal to $\EmbMor{\brc{0}}{\Int}$ \ie{}, if $\prn{\tilde{A}}^\Int \to
  \prn{\tilde{A}}^{\brc{0}} \times_{X^{\brc{0}}} X^\Int$ is an equivalence.
\end{lemma}
This formulation also makes plainer the fact that covariant families are the simplicial type theory
analogue of left fibrations in ordinary $\infty$-category theory~\parencite{joyal:2008,lurie:2009}.

Finally, using the characterization of covariance as an orthogonality condition, we are able to
prove the following:
\begin{lemma}
  \label{lem:covariance:flat-covariance-implies-simplicial}
  If $\DeclVar{A}{\GM}{X \to \Uni}$ is covariant then $A$ is simplicial \ie{}, it factors through
  $\Uni[\Simp] \to \Uni$.
\end{lemma}
\begin{proof}
  By \cref{ax:cubes-separate}, it suffices to show the following:
  \[
    \IsEquiv\prn{
      \Modify[\GM]{\Int^n \to \Sum{x : X} \IsSimp\prn{A\prn{x}}}
      \to
      \Modify[\GM]{\Int^n \to X}
    }
  \]
  In particular, we may assume that $X = \Int^n$ by restricting $A$ and so we hereafter also
  assume that $X$ is simplicial.

  With this in mind, $A$ factors through $\Uni[\Simp]$ if and only if $\Sum{x : X} A\prn{x}$ is
  simplicial. Next, we observe that $\Sum{x : X} A\prn{x}$ is simplicial if and only if the
  projection map $\Sum{x : X} A\prn{x} \to X$ is right orthogonal to
  $\prn{\Sum{i,j : \Int} i \le j \lor j \le i} \to \Int \times \Int$.

  Since $A$ and $X$ are both $\GM$-annotated, we may use \cref{ax:cubes-separate} again to reduce to
  showing that $\Sum{x : X} A\prn{x} \to X$ is $\GM$-orthogonal to
  $\Int^n \times \prn{\Sum{i,j : \Int} i \le j \lor j \le i} \to \Int^{n + 2}$.
  We will now argue that this map is in the left class generated by the inclusion $\brc{0} \to
  \Int^n$; we know that $\Sum{x : X} A\prn{x} \to X$ is right orthogonal to such maps by virtue of
  our assumption that $A$ was covariant.

  To this end, we consider the canonical inclusions
  $\brc{0} \to \Int^n \times \prn{\Sum{i,j : \Int} i \le j \lor j \le i}$
  and $\brc{0} \to \Int^{n + 2}$.
  Using the 3-for-2 property available for the left class of maps, to show that
  $\Sum{x : X} A\prn{x} \to X$ is right orthogonal to 
  $\Int^n \times \prn{\Sum{i,j : \Int} i \le j \lor j \le i} \to \Int^{n + 2}$,
  it suffices to show that (1) it is orthogonal to $\brc{0} \to \Int^{n + 2}$ and (2)
  it is orthogonal to 
  $\brc{0} \to \Int^n \times \prn{\Sum{i,j : \Int} i \le j \lor j \le i}$.
  The first claim is immediate from our assumption that $A$ is covariant.

  For the second claim, we note that
  $\brc{0} \to \Int^n \times \prn{\Sum{i,j : \Int} i \le j \lor j \le i}$ is the pushout of the maps
  $\brc{0} \to \Int^n \times \Int$ and $\brc{0} \to \Int^n \times \Delta^2$. We note that the latter
  is a retract of $\brc{0} \to \Int^n \times \Int^2$ and so both of these maps are orthogonal to
  $\Sum{x : X} A\prn{x} \to X$, again by our assumption that $A$ is covariant. The conclusion follows
  by the closure of left classes of maps under colimits.
\end{proof}

\subsection{Amazing covariance}
\label{sec:covariance:amazing}

We now refine our search from a universe of groupoids to a universe of \emph{covariant fibrations}.
That is, we wish to define some universe $\Space$ such that a map $A \to \Space$ corresponds (in
some sense) to a covariant fibration over $A$. Let us leave this correspondence imprecise for now
and consider the behavior of $\Space$.

In light of \cref{lem:covariance:covariant-implies-groupoid}, the points of $\Space$ will be
covariant over $\ObjTerm{}$ \ie{} groupoids. However, elements $f : \Int \to \Space$ will
become richer: they are covariant fibrations $B \to \Int$, therefore consist not only of a pair of
groupoids $B_0,B_1$ over $0$ and $1$, but also include a transport function $B_0 \to B_1$
(\cref{lem:covariance:transport}). Phrased differently, a homomorphism $F : \Int \to \Space$
contains an ordinary function $F\prn{0} \to F\prn{1}$.

Clearly this is a step towards directed univalence over $\Sum{A : \Uni} \Con{isGroupoid}\,A$, but it
is far from obvious how to define such a type $\Space$. In particular, while we have sketched how
behavior ought to differ between elements of $\Space$ compared with functions $\Int \to \Space$ and
so on, we cannot really cleanly divide elements of $\Space$ from functions into $\Space$ within type
theory. Within dependent type theory, every element of $\Space$ is formed to some context $\Gamma$
and if that $\Gamma$ happens contains a variable $i : \Int$, then this term will induce a function
$\Int \to \Space$.

There is an even more straightforward way to see why this causes a problem. Suppose we attempt to
define another subtype of $\Uni$ to isolate this universe of covariant fibrations
$\Sum{A : \Uni} \Con{isCov}\prn{A}$. A cursory inspection reveals this to be nonsense: being
covariant is not a property of $A$, it is a property of a family of types $A : X \to \Uni$. So in
this `definition', what exactly is $A$ covariant over?

It is here that modalities are vital: $A$ should be covariant with respect to the entire ambient
context. This is not something that can be expressed in standard type theory, but with the amazing
right adjoint to $\Int \to -$ we are able to define such a subtype.

\paragraph{Types covariant over $\Gamma$} We define a predicate on types $\IsACov : \Uni \to \Prop$
which encodes whether a type is covariant over the entire context following \textcite{riley:2024}.
We note that this predicate is a refinement of \textcite{licata:2018} which capitalizes on the
existence of the amazing right adjoint as a proper modality. The construction of this predicate
proceeds in two steps:
\begin{enumerate}
  \item We begin by observing that $\IsCovFam_X$ has the type $\prn{X \to \Uni}
    \to \Prop$ for each type $X : \Uni$. In particular,
    $\IsCovFam_\Int : \prn{\Int \to \Uni} \to \Prop$
  \item As $\IsCovFam_\Int$, we may apply \cref{lem:ttt:amazing-transpose} to obtain a function
    $\Uni \to \Prop_\Int$.
  \item Finally, we post-compose with the dependent version of $(-)_\Int$ to construct a function
    $\Uni \to \Prop$.
\end{enumerate}

All told, we obtain a predicate $\IsACov : \Uni \to \Prop$ which encodes whether a given type is covariant
over the entire context.

\begin{definition}
  A type is said to be \emph{amazingly covariant} when it satisfies $\IsACov$.
\end{definition}

We begin by substantiating the claim that $\IsACov\prn{A}$ implies that $A$ is truly covariant over
all variables in the context.

\begin{lemma}
  \label{thm:covariance:amazing-implies-ordinary}
  Given $F : X \to \Sum{A : \Uni} \IsACov\prn{A}$, the type family $F_0 = \Proj[1] \circ F$ is
  covariant.
\end{lemma}
\begin{notation}
  We will write $\Uni[\ACov]$ for the subtype $\Sum{A : \Uni} \IsACov\prn{A}$.
\end{notation}

\begin{proof}
  We must show $\IsCovFam\prn{F_0}$. First, we note that since $X \to \Uni$ being covariant
  implies that the composite $Y \to X \to \Uni$ is covariant for all $Y \to X$, we may reduce to the generic case
  where $X = \Uni[\ACov]$ and, in particular, where $F$ is $\GM$-annotated.

  Next, note that $\Uni[\ACov]$ fits into the following pullback diagram:
  \[
    \DiagramSquare{
      width = 4cm,
      height = 2cm,
      nw = \Uni[\ACov],
      ne = {\ObjTerm{}}_\Int,
      se = \Prop_\Int,
      sw = \Uni,
      south = \Transpose{\IsCov},
    }
  \]
  We therefore note that $X \to \Uni[\ACov]$ is equivalent to asking for a pair of maps $F_0 : X \to
  \Uni$ and $F_1 : X \to {\ObjTerm{}}_\Int$ such that the induced maps $X \to \Prop_\Int$ agree.
  Since $X \to \Uni[\ACov]$ along with the above pullback diagram consists only of $\GM$-annotated
  objects, we may therefore transpose to conclude that the induced maps $X^\Int \to \Prop$ agree.
  Unfolding, these maps are given as follows:
  \begin{gather*}
    \lambda f : X^\Int.\, \IsCov\prn{F_0 \circ f}
    \\
    \lambda f : X^\Int.\, \ObjTerm{}
  \end{gather*}
  Consequently, that these two maps agree amounts to a proof that $F_0$ is covariant, as required.
\end{proof}

We emphasize that in the above $\Uni[\ACov]$ does not ``know about'' $X$. In particular, this is a
subtype of $\Uni$ such that any map into this subtype induces covariant families.

Finally, the additional burden of being covariant over the context does not apply when working under
$\Modify[\GM]{-}$, a reflection of the fact that $\Modify[\GM]{A}$ is ``a proof of $A$ not depending
on the context.''
\begin{lemma}
  \label{lem:covariance:global-amazingness-degenerates} If $\DeclVar{X}{\GM}{\Uni}$ and
  $\DeclVar{A}{\GM}{X \to \Uni}$ then $\Modify[\GM]{\prn{x : X} \to \IsACov\prn{A\prn{x}}} =
  \Modify[\GM]{\IsCovFam\prn{A}}$.
\end{lemma}

\begin{proof}
  Since both $\IsACov$ and $\IsCovFam$ are propositions, it suffices to construct a bi-implication.
  First, let us suppose that $\DeclVar{z}{\GM}{\prn{x : X} \to \IsACov\prn{A\prn{x}}}$ holds.
  Applying the introduction rule for $\Modify[\GM]{-}$, we wish to show that $\IsCovFam\prn{A}$
  holds. Let us note that $z$ implies that $A : X \to \Uni$ factors through $\Uni[\ACov]$. That is,
  we have a diagram of the following shape (consisting of $\GM$-types and functions):
  
  \[
    \begin{tikzpicture}[diagram]
      \SpliceDiagramSquare{
        width = 3cm,
        height = 1.5cm,
        nw = \Uni[\ACov],
        ne = {\ObjTerm{}}_\Int,
        se = \Prop_\Int,
        sw = \Uni,
        south = \Transpose{\IsCov},
      }
      \node (A) [left = 3cm of sw] {$X$};
      \path[->] (A) edge node [below] {$A$} (sw);
      \path[->,exists] (A) edge (nw);
    \end{tikzpicture}
  \]

  By the naturality of transposition, we conclude that
  $\IsCovFam \circ A_* = \lambda \_.\,\ObjTerm{}$ as functions from $X^\Int \to \Prop$.
  Consequently, $A$ is covariant as required.

  For the reverse direction, suppose that $\DeclVar{z}{\GM}{\IsCovFam\prn{A}}$. It suffices to
  construct a (necessarily unique) function $\DeclVar{A_0}{\GM}{X \to \Uni[\ACov]}$ with an
  identification $\Proj \circ A_0 = A$. By the universal property of pullbacks along with the
  identification $\ObjTerm{} = {\ObjTerm{}}_\Int$, it therefore suffices to construct an
  identification between $\IsCovFam_\Int \circ A$ and $\lambda \_.\,{\ObjTerm{}}_\Int$. After
  transposing, we therefore must show that the map $X^\Int \to \Prop$ sending $x : \Int \to X$ to
  $\IsCovFam\prn{A \circ x}$ is equal to the map sending $x : \Int \to X$ to $\ObjTerm{}$. This,
  finally, follows immediately from our assumption that $\IsCovFam\prn{A}$ holds.
\end{proof}

\subsection{Closure properties of amazing covariance}
\label{sec:covariance:closure}

Given the strength of $\IsACov$, the reader may wonder how one ever proves that $\IsACov\prn{A}$
for any element $A : \Uni$. In this section, we give a partial answer by building up a stock of
amazingly covariant types. We shall see in \cref{sec:space} that these results undergird the closure
properties of our directed univalent universe. Our main result is the following:
\begin{lemma}
  \label{thm:covariance:acov-closure}
  In what follows, let us assume that $A,A_0,A_1 : \Uni$ and $B : A \to \Uni$.
  \begin{enumerate}
  \item If $\DeclVar{X}{\GM}{\Uni}$ then $\IsACov\prn{\Modify[\GM]{X}}$.
  \item If $i : \Int$ then $\IsACov\prn{i = 1}$.
  \item If $\IsACov\prn{A}$ and $a,b : A$ then $\IsACov\prn{a = b}$.
  \item If $\IsACov\prn{A}$ and $\prn{a : A} \to \IsACov\prn{B\prn{a}}$ then
    $\IsACov\prn{\Sum{a : A} B\prn{a}}$.
  \item If $\IsACov\prn{A_0}$, $\IsACov\prn{A_1}$ and $f,g : A_0 \to A_1$ then
    $\IsACov\prn{\Coeq\prn{f,g}}$.\footnote{Here $\Coeq\prn{f,g}$ denotes the coequalizer of $f,g$
      realized as a higher-inductive type~\parencite{hottbook}.}
  \end{enumerate}
  Moreover, $\IsACov$ is closed under $\Pi$-types provided modalities are used to
  manage the variance swap:
  \begin{enumerate}[resume]
  \item If $\DeclVar{C}{\OM}{\Uni}$ and $D : \MFn[\OM]{A}{\Uni}$ such that $\Modify[\OM]{\IsACov\prn{C}}$
    and $\prn{\DeclVar{c}{\OM}{C}} \to \IsACov\prn{D\prn{c}}$ then
    $\IsACov\prn{\prn{\DeclVar{c}{\OM}{C}} \to D\prn{c}}$.
  \end{enumerate}
\end{lemma}
We record a useful special case of (6) which follows from the involutive property of $\Modify[\OM]{-}$:
\begin{corollary}
  If $\DeclVar{X}{\GM}{\Uni}$, $B : X \to \Uni$ such that $\Prod{x : X} \IsACov\prn{B\prn{x}}$
  then $\IsACov\prn{\Prod{x : X} B\prn{x}}$.
\end{corollary}

We limit ourselves to proving three representative cases of the above theorem: (2), (4), (6). These
are indicative of the remaining cases (and those we have particular use for in
\cref{sec:space}).

\begin{lemma}
  If $i : \Int$ then $\IsACov\prn{i = 1}$.
\end{lemma}
\begin{proof}
  To prove this result, we shall switch to a more general goal,
  $\Modify[\GM]{\prn{i : \Int} \to \IsACov\prn{i = 1}}$, which can then be specialized to yield the
  original result.  Using \cref{lem:ttt:amazing-transpose}, it suffices to construct an element of
  $\Modify[\GM]{\prn{f : \Int \to \Int} \to \IsCovFam\prn{\lambda j.\,f\prn{j} = 1}}$

  Since we have no additional hypotheses in this proof, we may forget the $\Modify[\GM]{-}$ and assume
  $f : \Int \to \Int$. By \cref{ax:int}, $\Int$ is an h-set and so $\IsCovFam\prn{f\prn{j} = 1}$ is
  equivalent to showing that $f\prn{0} = 1$ implies that $f\prn{1} = 1$ \ie{} that $f$ is
  monotone. This is an immediate consequence of \cref{lem:ttt:phoa}.
\end{proof}

\begin{lemma}
  \label{lem:fibrations:sigma}
  If $\IsACov\prn{A}$ and $\prn{a : A} \to \IsACov\prn{B\prn{a}}$ then $\IsACov\prn{\Sum{a : A}
  B\prn{a}}$.
\end{lemma}
\begin{proof}
  As before, we begin by generalizing slightly and instead proving the following:
  \[
    \Modify[\GM]{
      \prn{\prn{A,B} : \Sum{A : \Uni[\ACov]} \Uni[\ACov]^A} \to \IsACov\prn{\Sum{a : A} B\prn{a}}
    }
  \]
  Let us begin by applying (the dependent version of) \cref{lem:ttt:amazing-transpose} such that it
  suffices to show that the following holds instead:
  \[
    \Modify[\GM]{
      \prn{\prn{A,B} : \Sum{A : \Uni[\ACov]^\Int} \prn{i : \Int} \to A\prn{i} \to \Uni[\ACov]}
      \to \IsCov\prn{\lambda i.\,\Sum{a : A\,i} B\prn{i,a}}
    }
  \]
  Since we have no additional assumptions, we may drop the $\GM$ and assume we are given $A :
  \Uni[\ACov]^\Int$ and $B : \prn{i : \Int} \to A\,i \to \Uni[\ACov]$. In light of
  \cref{thm:covariance:amazing-implies-ordinary}, we note that $A$ is covariant and, moreover, so
  too is $\lambda i.\,B\,i\,\prn{a\,i}$ for any $a : \prn{i : \Int} \to A\,i$. 

  In total then, we are reduced to proving the following: if $A : \Int \to \Uni$ and $B : \prn{i :
  \Int} \to A\prn{i} \to \Uni$ such that $\IsCov\prn{A}$ and $\IsCov\prn{\lambda
  i.\,B\,i\,\prn{a\,i}}$ then $\IsCov\prn{\lambda i.\,\Sum{a : A\,i} B\,i\,a}$. This statement is
  proven by \textcite[Proposition 6.2.1]{buchholtz:2023}.
\end{proof}

For clarity, we show the proof of the non-dependent version of (6). This is all that is required in
the next section and illustrates the core idea with less noise from handling indices.
\begin{lemma}
  If $\DeclVar{C}{\OM}{\Uni[\ACov]}$ and $D : \Uni[\ACov]$ then $\IsACov\prn{\Modify[\OM]{C} \to D}$.
\end{lemma}
\begin{proof}
  Following the previous two arguments, we will begin by proving this in a $\GM$-context to deal
  with $\IsACov$. That is, we first prove the following:
  \[
    \Modify[\GM]{
      \prn{\DeclVar{C}{\OM}{\Uni[\ACov]}}\prn{D : \Uni[\ACov]} 
      \to \IsACov\prn{\Modify[\OM]{C} \to D}
    }
  \]

  Arguing as in \cref{lem:fibrations:sigma}, we may use
  \cref{thm:covariance:amazing-implies-ordinary} and \cref{lem:ttt:amazing-transpose} to assume that
  we are given $\DeclVar{C}{\OM}{\Int \to \Uni}$ and $D : \prn{i : \Int} \to \Uni$
  such that (1) $\Modify[\OM]{\IsCov\prn{C}}$ and (2) $\IsCov\prn{D}$ and such that we
  must show (3) that $\IsCov\prn{\lambda i.\,\Modify[\OM]{C\prn{\neg i}} \to D\prn{i}}$.

  We begin by noting that $\Modify[\OM]{\IsCov\prn{C}}$ is equivalent to the following assumption:
  \[
    \prn{c_1 : \Modify[\OM]{C\prn{0}}} \to
    \IsContr\prn{\Sum{c : \prn{i : \Int} \to \Modify[\OM]{C\prn{\neg i}}} c\prn{1} = c_1}
  \]
  Unfolding our obligation, we must show that given $f_0 : \Modify[\OM]{C\prn{1}} \to {D\prn{0}}$ that
  the following type is contractible:
  \[
    \Sum{f : \prn{i : \Int} \to \Modify[\OM]{C\prn{\neg i}} \to D\prn{i}} f\prn{0} = f_0
  \]

  We first informally describe how one produces a center of contraction. Given $i : \Int$ and $c :
  \Modify[\OM]{C\prn{\neg i}}$, we use the assumption that $C$ is $\OM$-covariant as formulated
  above to construct a (unique) function $\bar{c} : \prn{j : \Int} \to \Modify[\OM]{C\prn{\neg
  \prn{i \land j}}}$ such that $\bar{c}\prn{1} = c$. We then observe that $\bar{c}\prn{0} :
  \Modify[\OM]{C\prn{1}}$ and so $d_0 = f_0\prn{\bar{c}\prn{0}} : D\prn{0, \bar{c}\prn{0}}$.
  Extending this $d_0$ to a line along $D\prn{i \land -}$ gives the line $\bar{d} : \prn{j : \Int}
  \to D\prn{j \land i}$ and we choose $f\prn{i,c} = f_0\prn{\bar{d}\prn{1}} : \prn{c :
  \Modify[\OM]{C\prn{\neg i}}} \to D\prn{i,c}$. If $i = 0$, $\bar{c}$ and $\bar{d}$ are canonically
  equal to constant functions and these identifications combine to produce a path $p : f\prn{0} =
  f_0$.

  If we write $\widebar{\Coe}_{C}$ for the ``backwards'' coercion function induced by
  $\Modify[\OM]{\IsCov\prn{C}}$, we may describe $f$ symbolically as follows:
  \[
    f = \lambda i\,c.\Coe_{D\prn{i \land -}}\prn{f_0\prn{\widebar{\Coe}_{C\prn{\neg i \land -}}\prn{c}}}
  \]

  Suppose now we are given $g : \prn{i : \Int} \to \Modify[\OM]{C\prn{\neg i}} \to D\prn{i}$ along
  with $p : g\prn{0} = f_0$. We must then show that $\prn{f,p} = \prn{g,q}$. Let us consider the
  following $H$:
  \begin{align*}
    &H : \prn{i\,j : \Int} \to \Modify[\OM]{C\prn{\neg i}} \to D\,i
    \\
    &H =
    \lambda i\,j,c.\Coe_{D\prn{(i \land j) \lor - \land i}}
    \prn{g\,(i \land j)\,\prn{\widebar{\Coe}_{C\prn{\neg (i \land j) \land - \lor \neg i}}\prn{c}}}
  \end{align*}

  We may construct a function $r : \prn{j : \Int} \to H\prn{0,j} = f_0$ using $p$ and $q$. Moreover,
  we may identify $\prn{H\prn{-,0}, r\prn{0}}$ with $\prn{f,p}$ and $\prn{H\prn{-,1}, r\prn{1}}$
  with $\prn{g,q}$. Finally, since the type $\prn{i : \Int} \to \Modify[\OM]{C\prn{\neg i}} \to
  D\prn{i}$ is a groupoid~\parencite[Lemma 1.26]{rijke:2020} and groupoids are closed under
  $\Sigma$-types, $\lambda j.\,\prn{H\prn{-,j}, r\prn{j}}$ induces the required identification
  between $\prn{f,p}$ and $\prn{g,q}$ as required.
\end{proof}

These proofs exhibit proof strategies that are common when working with $\IsACov$ in \TTT{}:
either reducing to a generic global case where various modalities can be simplified or performing
several small modal manipulations and then applying standard and non-modal arguments. They are
also very similar to the construction of fibrancy structures in \textcite{weaver:2020}. The major
difference between the proofs given in \opcit{} stems from the fact that our constructions take
place in a univalent type theory. Consequently, our coercion operators have a simpler type, but we
must show that they are unique up to a contractible choice rather than merely having to construct
some inhabitant.

\section{The directed univalent universe}
\label{sec:space}

With our preliminary work on amazing covariance in place, we are now in a position to define our
directed univalent universe of groupoids $\Space$ and establish its core properties. We begin with
the (now short) definition of $\Space$:
\begin{definition}
  We define $\Space$ to be $\Sum{A : \Uni}{\IsACov\prn{A}}$.
\end{definition}

We note that $\Space$ can be fully characterized without reference to
$\ACov$ as a corollary of \cref{lem:covariance:global-amazingness-degenerates}:%
\footnote{Theoretically, every result about $\Space$ can be proven using this
characterization. We will not endeavor to do so and instead optimize for more readable proofs.}

\begin{lemma}\label{lem:space:space-characterization}
  If $\DeclVar{A}{\GM}{X \to \Uni}$ then $A$ factors
  through $\Space$ if and only if it is covariant. In other words, $\Space$ is the base of the
  universal covariant family of simplicial types.
\end{lemma}

\begin{corollary}
  \label{cor:space:flat-groupoid}
  $\DeclVar{A}{\GM}{\Uni}$ factors through $\Space$ if and only if it is a groupoid.
\end{corollary}

Applying \cref{lem:covariance:flat-covariance-implies-simplicial}, we conclude that $\Space \to
\Uni$ factors through $\Uni[\Simp]$. In other words, all amazingly covariant families of groupoids
are automatically simplicial. Moreover, by \cref{thm:covariance:acov-closure} along with the closure
results from \textcite{rijke:2020}, we conclude:

\begin{lemma}
  \label{thm:space:space}
  As a subtype of $\Uni$, $\Space$ is (1) univalent (2) contains only simplicial types (3) closed
  under dependent sums, equality, and $i = 1$ (4) closed under the two modalized forms of
  $\Pi$-types indicated by \cref{thm:covariance:acov-closure}.
\end{lemma}

Thus, we already have established that $\Space$ is a subuniverse of $\Uni$ spanned by
groupoids. What remains is to prove directed univalence \ie{}, to characterize $\Int \to \Space$. To
this end, we will first prove two important lemmas for constructing elements of $\Space$. With these
in place, we shall show that $\Space$ is not only closed under various connectives, but also
simplicial, Segal, Rezk, and directed univalent. Our main result can be summed up as follows
\begin{theorem}
  $\Space$ is a directed-univalent category.
\end{theorem}

\begin{remark}
  While it is not helpful for establishing the above theorem, the following observation is
  helpful for seeing that the construction of $\Space$ using cubical spaces must yield the same
  results as a more standard argument within simplicial spaces and, in fact, must produce the
  standard $\infty$-category of $\infty$-groupoids.

  Suppose that $\Space'$ was constructed as the base of the universal covariant fibration among
  \emph{simplicial} types, we can then argue that $\Space'$ is also necessarily equivalent to the
  base of the universal covariant fibration among cubical spaces. First, observe that the universal
  fibration $\Space'_\bullet \to \Space'$ is covariant when viewed as a map of cubical spaces:
  covariance is equivalent to the map $\prn{\Space'_\bullet}^\Int \to \Space'_\bullet
  \times_{\Space'} \prn{\Space'}^\Int$ being an equivalence and simplicial spaces are a full
  subcategory of cubical spaces closed under limits and exponentials. Next, we see that this entails
  the existence of a classifying map $\Space' \to \Space$ and it is routine to calculate that this
  map induces an equivalence $\Hom{\Int^n}{\Space'} \Equiv \Hom{\Int^n}{\Space}$ since $\Int^n$ is
  simplicial. Consequently, $\Space \Equiv \Space'$. In fact, a nearly identical argument shows that
  $\Space'$ must in turn agree with the base of a  covariant family universal among covariant
  families of complete Segal spaces (\ie{}, $\infty$-categories) as well---again assuming such a
  thing exists.

  Consequently, if one assumes the ordinary statement of straightening--unstraightening for the
  $\infty$-category of $\infty$-groupoids, then the base of the universal covariant family in
  cubical spaces must agree with the standard $\infty$-category of $\infty$-groupoids. While not a
  satisfactory method of constructing $\Space$, this does show that our construction must yield the
  expected result.
\end{remark}

\subsection{The two key lemmas}

Before we can prove that $\Space$ is directed univalent, we require a better understanding of when
two maps $\Int \to \Space$ are equivalent. In particular, suppose we are given
$f,g : \Int \to \Space$. We already know that $\Space$ is univalent as a subtype of $\Uni[\Simp]$
and so $f$ and $g$ are equal when there is an equivalence $\alpha : \prn{i : \Int} \to f\prn{i} \to
g\prn{i}$. Accordingly, it suffices to find conditions to establish that $\alpha\prn{i}$ is an
equivalence for each $i : \Int$. Our first result shows that this holds everywhere if it holds at
$0$ and $1$. In other words, to check that a natural transformation $\alpha$ is an equivalence, it
suffices to check that it is an equivalence at each object. We prove a slight generalization of this
result which applies to any $\Delta^\ell$ rather than just $\Delta^1$.

\begin{notation}
  We denote $\prn{1, \dots, 1, 0, \dots 0} : \Delta^\ell$ with $k$ copies of $1$ followed by
  $\ell - k$ copies of $0$ by $\bar{k}$.
\end{notation}

\begin{lemma}
  \label{lem:space:extension}
  Fix $\DeclVar{\ell}{\GM}{\Nat}$ and suppose that $\Mor[f,g]{\Delta^\ell}{\Space}$ and
  $\alpha : \prn{\delta : \Delta^\ell} \to f\,\delta \to g\,\delta$ then $\alpha$ is invertible if
  and only if $\Mor[\alpha\,\bar{k}]{f\,\bar{k}}{g\,\bar{k}}$ is invertible for all $k \le \ell$.
\end{lemma}
\begin{proof}
  We begin by generalizing to apply \cref{ax:cubes-separate}. To this end, fix the following global
  types:
  \begin{align*}
    &X =
    \Sum{F\,G : \Delta^\ell \to \Space}\Sum{\alpha : \prn{\delta:\Delta^\ell} \to F\,\delta \to G\,\delta}
    \prn{k : \Nat_{\le \ell}} \to \IsEquiv\prn{\alpha\,\bar{k}}
    \\
    &Y =
    \Sum{F\,G : \Delta^\ell \to \Space}\Sum{\alpha : \prn{\delta:\Delta^\ell} \to F\,\delta \to G\,\delta}
    \Prod{\delta : \Delta^{\ell}}\IsEquiv\prn{\alpha\prn{\delta}}
  \end{align*}
  It suffices to show that the forgetful map $Y \to X$ is an equivalence and so, by
  \cref{ax:cubes-separate}, we must show that for each $\DeclVar{n}{\GM}{\Nat}$ the map
  $\Modify[\GM]{\Int^n \to Y} \to \Modify[\GM]{\Int^n \to X}$ is an equivalence. For clarity, we write
  $\Gamma = \Int^n$ and $\Gamma' = \Int^n \times \Delta^{\ell}$ in what follows.

  We now unfold this slightly. Fix $\DeclVar{F,G}{\GM}{\Gamma' \to \Space}$ along with
  $\DeclVar{\alpha}{\GM}{\prn{\prn{v,\delta} : \Gamma'} \to F\prn{v,\delta} \to G\prn{v,\delta}}$ and
  $\DeclVar{e}{\GM}{\prn{v : \Gamma}\prn{k : \Nat_{\le \ell}} \to
    \IsEquiv\prn{\alpha\prn{v,\bar{k}}}}$. We must show the following:
  \[
    \Modify[\GM]{\prn{\prn{v,\delta} : \Gamma'} \to \IsEquiv\prn{\alpha\prn{v,\delta}}}
  \]

  We can reorient $F,G$ as global families $\Mor[\pi_F,\pi_G]{\tilde{F},\tilde{G}}{\Gamma'}$. That
  both $F,G$ factor through $\Space$ implies that both projections are both covariant fibrations
  and, therefore, orthogonal to the maps $\brc{0} \to \Int^m$ for any
  $\DeclVar{m}{\GM}{\Nat}$. Note, too, that from this viewpoint, $\alpha$ is a map
  $\tilde{\alpha} : \tilde{F} \to \tilde{G}$ over $\Gamma'$ such that pulling back along
  $\prn{\ArrId{},\bar{k}} : \Gamma \to \Gamma'$ induces an equivalence. We must show that
  $\tilde{\alpha}$ is an equivalence.

  By another application \cref{ax:cubes-separate}, to show that $\tilde{\alpha}$ is an equivalence
  we must show it induces an equivalence
  $\Modify[\GM]{\Int^m \to \tilde{F}} \Equiv \Modify[\GM]{\Int^m \to \tilde{G}}$. By orthogonality, we
  note that
  $\Modify[\GM]{\Int^m\to\tilde{F}} \Equiv \Modify[\GM]{\tilde{F} \times_{\Gamma'} \prn{\Int^m \to
      \Gamma'}}$.  Consequently, it suffices to show that the following map is an equivalence:
  \[
    \Modify[\GM]{\tilde{F} \times_{\Gamma'} \prn{\Int^m \to \Gamma'}} \to
    \Modify[\GM]{\tilde{G} \times_{\Gamma'} \prn{\Int^m \to \Gamma'}}
  \]
  We may refactor this using the various properties of $\Modify[\GM]{-}$ to obtain the following
  equivalent map:
  \[
    \Sum{\DeclVar{v}{\GM}{\Int^m \to \Gamma}}
    \Sum{\DeclVar{\theta}{\GM}{\Int^m \to \Delta^{\ell}}}
    \Modify[\GM]{F\prn{v\prn{\vec{0}},\theta\prn{\vec{0}}}}
    \to
    \Sum{\DeclVar{v}{\GM}{\Int^m \to \Gamma}}
    \Sum{\DeclVar{\theta}{\GM}{\Int^m \to \Delta^{\ell}}}
    \Modify[\GM]{F\prn{v\prn{\vec{0}},\theta\prn{\vec{0}}}}
  \]
  Finally, $\theta\prn{\vec{0}}$ is an element of $\Modify[\GM]{\Delta^{\ell}}$ and is therefore
  equal to $\bar{k}$ for some $k$ by \cref{ax:global-points}. For any $k$, the map is an equivalence
  as it is derived from $\alpha$ and our conclusion follows.
\end{proof}
\begin{remark}
  \textcite{weaver:2020} axiomatized their cobar modality to formulate and postulate a special
  case of this lemma (their \emph{equivalence axiom}). In our case, no such steps are required as
  this result follows from \cref{ax:cubes-separate}. This is not to say that the cobar construction
  plays no role in our setting: it is used in \cref{sec:model} implicitly as
  \textcite{shulman:2019} uses it to characterize the injective fibrations we use to model types.
\end{remark}

To ensure that elements of $\Space$ are indeed groupoids, $A : \Uni$ lands in the subtype $\Space$
only when it is simplicial in addition to being amazingly covariant. Often, it is easiest to do
this by proving that $A$ is amazingly covariant and then applying $\Simp$ to $A$ to obtain a
simplicial type. In order for this to be possible, however, we must know that applying $\Simp$ to an
amazingly covariant type results in an amazingly covariant type. The next lemma proves (a
generalization of) this fact.

Let us note that the canonical maps $\pi : \Space \to \Uni$ and
$\DepSimp : \Simp \Uni \to \Uni$ induce a map $\DepSimp \Proj : \Simp \Space \to \Uni$. Showing
that $\Simp A$ is amazingly covariant if $A$ is amazingly covariant corresponds to showing that
$\DepSimp \Proj \circ \eta$ factors through $\Space$. We prove this by proving the following
stronger result:
\begin{lemma}[Simplicial exchange]
  \label{lem:space:simplicial-exchange}
  $\DepSimp \Proj : \Simp \Space \to \Uni$ factors through $\Space$.
\end{lemma}
\begin{proof}
  Given that the composite $\DepSimp \Proj : \Simp \Space \to \Uni$ mentions no free variables, it
  suffices by \cref{lem:covariance:global-amazingness-degenerates} to show that $\DepSimp \Proj$ is
  covariant.

  For concision, we write $X = \Simp \Space$ and $\tilde{X}$ for $\Sum{A : X} \DepSimp A$. We
  must show that the map given by evaluating at $0$ induces an equivalence between
  $\Int \to \tilde{X}$ and $\prn{\Int \to X} \times_{X} \tilde{X}$. Using \cref{ax:cubes-separate}
  along with the observation that these types are all simplicial, it suffices to show that the
  following map is an equivalence for all $\DeclVar{n}{\GM}{\Nat}$:
  \[
    \Modify[\GM]{\Delta^n \times \Int \to \tilde{X}}
    \to
    \Modify[\GM]{
      \Delta^n
      \to
      \prn{\Int \to X} \times_{X} \tilde{X}
    }
  \]
  In other words, we must show that $\Delta^n \times \brc{0} \to \Delta^n \times \Int$ is
  \emph{globally} orthogonal to $\tilde{X} \to X$.

  We now argue that to prove this, it suffices to show that $\tilde{X} \to X$ is globally orthogonal
  to $\brc{0} \to \Delta^k$ for all $k$. Let us assume that this condition holds for the moment and
  show that it suffices to establish orthogonality with respect to
  $\Delta^n \times \brc{0} \to \Delta^n \times \Int$. Considering the sequence of maps 
  $\brc{0} \to \Delta^n \times \brc{0} \to \Delta^n \times \Int$, a standard 3-for-2 argument shows
  that it suffices to show that $\brc{0} \to \Delta^n \times \Int$ is orthogonal to
  $\tilde{X} \to X$.

  To this end, notice that $\Delta^n \times \Int$ is a
  subtype of $\Int^{n + 1}$ spanned by tuples $\prn{v_1, \dots, v_n, w}$ where
  $v_1 \ge \dots \ge v_n$. Observing that both $\tilde{X}$ and $X$ are simplicial, it follows that
  a map $\Delta^n \times \Int \to X$ or $\Delta^n \times \Int \to \tilde{X}$ is the same as a map
  from the following subtype of $\Int^{n + 1}$:
  \begin{equation*}
    \Compr{
      \prn{v_1, \dots, v_n, w}
    }{
      \prn{w \ge v_1 \ge \dots \ge v_n}
      \lor
      \prn{v_1 \ge w \ge \dots \ge v_n}
      \lor
      \dots
      \lor
      \prn{v_1 \ge \dots \ge v_n \ge w}
    }
  \end{equation*}

  Let us denote the condition $v_1 \ge \dots \ge v_k \ge w \ge v_{k + 1} \ge v_n$ by $\Phi_{n,k}$
  with $\Phi_{n,0}$ being $w \ge v_1 \ge \dots \ge v_n$. We wish to show that for any element
  $\DeclVar{\tilde{x}}{\GM}{\tilde{X}}$ and map $\DeclVar{f}{\GM}{\Delta^n \times \Int \to X}$ along
  with a path $\DeclVar{p}{\GM}{\pi\,\tilde{x} = f\prn{0, \dots, 0}}$ that there exists a unique
  extension of $f$ to $\tilde{f}$. We may phrase this as constructing an element of the following type:
  \[
    \IsContr\prn{
      \prn{t : \Delta^n \times \Int} \to
      \Sum{\tilde{y} : \tilde{X}} \Sum{q : \pi\prn{\tilde{y}} = f\prn{t}}
      t = 0 \to \Sum{r : \tilde{x} = \tilde{y}} p = \pi\prn{r} \bullet q
    }
  \]

  Since we are therefore constructing a map out of a series of disjunctions (and constructing an
  identification of such maps) it suffices to show that this proposition is inhabited for each
  subtype $\Compr{\prn{v_1, \dots, v_n, w}}{\Phi_{n,k_0} \land \dots \land \Phi_{n,k_i}}$ for each
  sequence $i$ and $k_0 \le k_1 \dots \le k_i$. Calculation reveals that the intersection
  $\Phi_{n,k_0} \land \dots \land \Phi_{n,k_i}$ is necessarily a sub-simplex of $\Delta^n \times
  \Int$. Consequently, unfolding definitions, constructing such a map on this subtype is precisely
  the same as showing that $\tilde{X} \to X$ is globally orthogonal $\brc{0} \to \Delta^{k}$ for
  some $k$, our assumption.

  All told then, it suffices to show the following canonical map is an equivalence:
  \[
    \Modify[\GM]{\Delta^m \to \tilde{X}}
    \to
    \Modify[\GM]{
      \prn{\Delta^m \to X} \times_{X} \tilde{X}
    }
  \]

  By \cref{ax:simplicial-stability}, we may ``remove the $\Simp$'' from $X$ and $\tilde{X}$ and so
  this type is equivalent to the following:
  \[
    \Modify[\GM]{\Delta^m \to \Sum{A : \Space} A}
    \to
    \Modify[\GM]{
      \prn{\Delta^m \to \Space} \times_{\Space} \prn{\Sum{A : \Space} A}
    }
  \]
  This, finally, is an equivalence because $\prn{\Sum{A : \Space} A} \to \Space$ is
  covariant (\cref{thm:covariance:amazing-implies-ordinary}).
\end{proof}
\begin{corollary}
  \label{cor:space:coeq}
  $\Space$ is closed under coequalizers in $\Uni[\Simp]$.
\end{corollary}
\begin{proof}
  By \cref{thm:covariance:acov-closure}, $\Uni[\ACov]$ is closed under coequalizers $\Coeq\prn{f,g}$
  and so \cref{lem:space:simplicial-exchange} ensures the $\Simp \Coeq\prn{f,g}$ lands in $\Space$
  as well. By \textcite{rijke:2020}, this is the coequalizer in $\Uni[\Simp]$.
\end{proof}

\subsection{\texorpdfstring{$\Space$}{S} is directed univalent, Segal, Rezk, and simplicial}

We are now able to show that $\Space$ satisfies all the desired properties for a universe of
groupoids. We begin by showing that we have, at last, constructed a directed univalent universe.

First, we note that \cref{def:intro:dua} merely states that there is some isomorphism between two
types. We are already in a position to construct one of these two maps:
\begin{lemma}
  There is a function $\MorToFun$ from $\Int \to \Space$ to $\Sum{A\,B : \Space} A \to B$.
\end{lemma}
\begin{proof}
  Given $F : \Int \to \Space$, by \cref{thm:covariance:amazing-implies-ordinary} this induces a
  covariant family $F_0 : \Int \to \Uni$. We then define
  $\MorToFun\prn{F} \defeq \prn{F_0\,0,F_0\,1,\Coe_F}$ where the last component is induced by
  \cref{lem:covariance:transport}.
\end{proof}

\begin{theorem}[Directed univalence]
  \label{thm:space:dua}
  The function $\MorToFun$ is an equivalence.
\end{theorem}

Prior to proving this result, we will construct a putative inverse to $\MorToFun$.

\begin{definition}
  Given $A,B : \Space$ and $f : A \to B$, $\Glue\prn{A,B,f} : \Int \to \Space$ is
  $\lambda i.\,\Sum{b : B} i = 0 \to f^{-1}\prn{b}$.
\end{definition}

$\Glue$ is the directed version of the glue type from cubical type
theory~\parencite{cohen:2017,sattler:2017} and is inspired directly from the construction used in
\textcite{weaver:2020} in their construction of a directed univalent universe.\footnote{Not to be
confused with \emph{Artin gluing} from categorical logic.} In our case, we have no need to add it as
a primitive in our setting: this was necessary in (bi)cubical type theory to achieve certain
definitional equalities, but we are pervasively working up to equivalence. We note that
$\Glue\prn{f}$ factors through $\Space$ by virtue of (2--4,6) of \cref{thm:covariance:acov-closure}
along \cref{ax:op-of-int-is-int} which ensures that $\Modify[\OM]{\neg j = 1} = \prn{j = 0}$.
We also record a few elementary fact about $\Glue$ below:

\begin{lemma}
  \label{lem:space:glue-facts}
  Given $A,B,f$ as above, $\Glue\prn{A,B,f}\,0 = A$, $\Glue\prn{A,B,f}\,1 = B$, and
  $\Coe_{\Glue\prn{A,B,f}} = f$.
\end{lemma}

We now return to the proof of \cref{thm:space:dua}.
\begin{proof}[Proof of \cref{thm:space:dua}]
  We will prove that $\Glue$ forms a quasi-inverse to $\MorToFun$ and thereby conclude that
  $\MorToFun$ is an equivalence. We must therefore prove (1) $\MorToFun \circ \Glue = \ArrId{}$ and
  (2) $\Glue \circ \MorToFun = \ArrId{}$. (1) follows from direct calculation and
  \cref{lem:space:glue-facts}, so we will detail only (2).

  Suppose we are given $F : \Int \to \Space$. We must show that $F = \Glue\prn{\MorToFun\prn{F}}$ or
  equivalently, using the fact that $\Space$ is univalent, that there is an equivalence
  $\alpha : \prn{i : \Int} \to F\prn{i} \Equiv \Glue\prn{\MorToFun\prn{F}}\,i$. To prove this, we
  will begin by constructing $\alpha$ and then use \cref{lem:space:extension} to reduce to checking
  that $\alpha$ is an equivalence at $0$ and $1$.  It is helpful to do this in stages and so we
  begin by supposing $i : \Int$ and $z : F\prn{i}$ and define $\alpha$ as follows for some $X$ and
  $Y$ to be determined:
  \[
    \alpha\,i\,z = \prn{X : F\prn{1}, Y : i = 0 \to \Coe_{F}^{-1}\prn{X}}
  \]
  We will construct $X$ and $Y$ separately.

  We can substantiate $X$ immediately: $\Coe_{F\prn{- \lor i}} : F\prn{i} \to F\prn{1}$ and so we
  choose $X \defeq F\prn{- \lor i}\,z$. This refines the type of $Y$ to $i = 0 \to
  \Coe_F^{-1}\prn{\Coe_{F\prn{- \lor i}}\,z}$. Assume $\phi : i = 0$ so that it suffices to define
  $Y.1 : F\prn{0}$ and $Y.2 : \Coe_F\,Y.1 = \Coe_{F\prn{- \lor i}}\,z$. Using $\phi$, we may suppose
  that $z : F\prn{0}$ and that the type of $Y.2$ is $\Coe_F\,Y.1 = \Coe_F\,z$
  (since $0 \lor - = \ArrId{}$). After this, $Y.1 \defeq z$ and $Y.2 \defeq \Refl$ suffices.

  Finally, it is now straightforward to check that $\alpha\,0$ and $\alpha\,1$ are equivalences
  using \cref{lem:space:glue-facts}.
\end{proof}

The proof that $\Space$ is Segal is very similar to the proof of directed univalence, though not
quite a consequence of it. Since the proof is similar to \cref{thm:space:dua}, we provide only a
sketch.
\begin{lemma}
  \label{thm:space:segal}
  $\Space$ is Segal.
\end{lemma}
\begin{proof}[Proof sketch.]
  We must show that $\prn{\Delta^2 \to \Space} \to \prn{\Lambda^2_1 \to \Space}$ is an
  equivalence. We begin by noting that the codomain can be rewritten with \cref{thm:space:dua} as
  $T = \Sum{A\,B\,C : \Space} A \to B \times B \to C$. We only need to show that the
  forgetful map from $\prn{\Delta^2 \to \Space} \to T$ is an equivalence.

  This argument proceeds along the same lines as \cref{thm:space:dua} where we replace
  $\Int$ with $\Delta^2$: we introduce a variant of $\Glue$ which glues together three spaces along
  two maps and show that this procedure induces a quasi-inverse to the forgetful map
  $\prn{\Delta^2 \to \Space} \to T$. It is here that we require \cref{lem:space:extension} with
  $\ell = 2$ rather than $\ell = 1$.
\end{proof}
\begin{corollary}
  \label{cor:space:comp}
  Composition of the morphisms in $\Space$ is realized by ordinary function composition.
\end{corollary}
In particular, an invertible morphism corresponds via \cref{thm:space:dua} to an
equivalence. Combining this with ordinary univalence, we obtain:
\begin{corollary}
  $\Space$ is Rezk.
\end{corollary}

Our final result is that $\Space$ lands in the subuniverse of simplicial types.
\begin{lemma}
  $\Space$ is simplicial.
\end{lemma}
\begin{proof}
  By \textcite[Lemma 1.20]{rijke:2020}, it suffices to show that $\eta : \Space \to \Simp \Space$ has a
  retraction. By univalence, the composite of $\eta : \Uni[\Simp] \to \Simp \Uni[\Simp]$ followed by
  $\DepSimp : \Simp \Uni[\Simp] \to \Uni[\Simp]$ is the identity and so it suffices to show that
  both these maps restrict to $\Space$. That is, it suffices to show that
  $\DepSimp \circ \Proj : \Simp \Space \to \Uni$ factors through $\Space$. This is an immediate
  consequence of \cref{lem:space:simplicial-exchange}.
\end{proof}

We conclude by noting a few of the categorical properties $\Space$ enjoys:
\begin{lemma}
  $\Space$ is finitely complete and finitely cocomplete and satisfies \emph{descent}~\parencite[Chapter
  2]{rijke:phd}.
\end{lemma}
\begin{proof}[Proof Sketch]
  Finite completeness and cocompleteness are an immediate consequence of
  \cref{thm:space:space,cor:space:coeq} along with \cref{thm:space:dua} which implies that a \eg{},
  categorical limit in $\Space$ is an ordinary \HOTT{} limit of groupoids. To prove the descent
  properties, we must show that various limits and colimits commute appropriately. However, by
  \cref{thm:space:dua} once more, this is an immediate consequence of the fact that limits and
  colimits in \HOTT{} enjoy descent~\parencite{rijke:phd}.
\end{proof}

\section{Consequences of a directed univalent universe}
\label{sec:applications}

We now reap the rewards of our efforts in constructing $\Space$ and give a brief tour of the
consequences of this type. We show how directed univalence may be used to prove free theorems and
substantiate the structure homomorphism principle. We also use it to construct various
foundational example categories and lay the groundwork for the development of \emph{higher algebra}
within \TTT{}.

\subsection{Free theorems from naturality}

Directed univalence allows us to make a precise link between familiar parametricity
arguments~\parencite{wadler:1989} with the categorical naturality arguments that helped motivate
them. In particular, directed univalence implies that a function
$\alpha : \prn{A : \Space} \to F\prn{A} \to G\prn{A}$ is natural:
\begin{lemma}
  \label{thm:consequences:natural}
  If $F_0,F_1 : \Space \to \Space$ and $\alpha : \prn{A : \Space} \to F_0\prn{A} \to F_1\prn{A}$ then
  $\alpha\prn{B} \circ F_0(f) = F_1(f) \circ \alpha\prn{A}$ for any $f : A \to B$.
\end{lemma}
\begin{proof}
  Fix $A,B : \Space$ along with $f : A \to B$ and denote the corresponding morphism
  $G : \Int \to \Space$. Note that $\alpha \circ G$ is then a function
  $\prn{i : \Int} \to F_0\prn{i} \to F_1\prn{i}$. Applying \cref{thm:space:dua} once more, we note
  that $\alpha\prn{G\prn{i}} : F_0\prn{i} \to F_1\prn{i}$ is a morphism in $\Space$ for every
  $i$. Accordingly, $\alpha \circ G$ is equivalent to some $s : \prn{i\,j : \Int} \to H\,i\,j$ for
  some $H$ where $H\,i\,0 = F_0\,i$ and $H\,i\,1 = F\,i$. We visualize $H$ as:
  \[
    \DiagramSquare{
      nw = F_0\,0,
      sw = F_1\,0,
      ne = F_0\,1,
      se = F_1\,1,
      north = F_0,
      south = F_1,
      west = \alpha\prn{G\,0},
      east = \alpha\prn{G\,1},
      width = 4cm,
    }
  \]
  This commuting square is equivalently an equality between the composites $F_1$ and
  $\alpha\prn{G\,0}$ and $\alpha\prn{G\,1}$ and $F_0$. The conclusion then follows from
  \cref{cor:space:comp}.
\end{proof}

\polyid
\begin{proof}
  Fix $A : \Space$ and suppose we are given $a : A$. Applying \cref{thm:consequences:natural} to $f$
  and $\lambda \_.\,a$, we conclude that $f\,A\,\prn{a\,\star} = a\prn{f\,\ObjTerm{}\,\star}$. Since
  $f\,\ObjTerm{}\,\star = \star$ by the $\eta$ principle of $\ObjTerm{}$, $f = \lambda A\,a.\,a$.
\end{proof}

Nothing limits us to considering only operations $\Space \to \Space$. The same techniques scale to
multi-argument operations such as $\Space \times \Space \to \Space$ or even mixed-variance
operations such as $\Modify[\OM]{\Space} \times \Space \to \Space$:
\begin{lemma}
  If $\alpha : \prn{A\,B : \Space} \to A \times B \to A$ then $\alpha = \Proj[1]$.
\end{lemma}
\begin{lemma}
  If $\DeclVar{A,B}{\GM}{\Space}$ and
  $\alpha : \prn{\DeclVar{C}{\OM}{\Space}} \to A^{\Modify[\OM]{C}} \to B^{\Modify[\OM]{C}}$ then
  $\alpha = \lambda\_g.f \circ g$ for some $f : A \to B$.
\end{lemma}

This methodology highlights the \emph{limitations} of naturality as a facsimile for parametricity:
for operations whose parameters are not used strictly co- or contravariantly, directed univalence
does not provide any free theorems. We leave it to future work to consider alternative universes of
correspondences~\parencite{ayala:2020} and what parametricity arguments they might provide.

\subsection{Full subcategories of \texorpdfstring{$\Space$}{S}}

A large number of important categories can be described as a \emph{full subcategories} of $\Space$.
To do this, we must first show how to obtain full subcategories inside of \TTT{}. Recall that a full
subcategory of $C_0$ of a category $\DeclVar{C}{\GM}{\Uni}$ is a category $C_0$ where objects
are a subset of those in $C$ but the morphisms and all the higher cells agree. In other words, a
full subcategory is described by a predicate $\MFn[\GM]{C}{\Prop_\Simp}$ which picks out those
objects which land in $C_0$.
\begin{definition}
  Given $\DeclVar{\phi}{\SM}{\MFn[\GM]{C}{\Prop_\Simp}}$, the resulting full subcategory
  $C_\phi$ is $\Sum{c : C} \Modify[\SM]{\phi\prn{c^\eta}}$.\footnote{In practice, $\phi$ will be
  $\GM$-annotated.}
\end{definition}

Here we for the first time have occasion to explicitly use the right adjoint $\SM$ to $\GM$. Let us
note that $C_\phi$ is a category because (1) categories are closed under dependent sums and (2)
$\Modify[\SM]{\phi\prn{c^\eta}}$ is a groupoid. Furthermore, we can prove that $C_\phi$ is
actually a full subcategory:

\begin{lemma}
  \label{thm:consequences:fullness} Given $C$ and $\phi$ as above, if $a,b : C_\phi$ then
  $\Hom[C_\phi]{a}{b} \Equiv \Hom[C]{\Proj[1]\,a}{\Proj[1]\,b}$.
\end{lemma}
\begin{proof}
  Unfolding definitions, it suffices to show the following two propositions are equivalent for all
  $f : \Int \to C$:
  \begin{gather*}
    \Modify[\SM]{\phi\prn{f^\eta\prn{0}}} \times \Modify[\SM]{\phi\prn{f^\eta\prn{1}}}
    \\
    \prn{i : \Int} \to \Modify[\SM]{\phi\prn{f^\eta\prn{i^\eta}}}
  \end{gather*}

  Since $\GM \Adjoint \SM$, we may use \cref{lem:ttt:transpose} to replace the second proposition
  with 
  $\Modify[\SM]{\prn{\DeclVar{i}{\GM}{\Int}} \to \phi\prn{f^\eta\prn{i}}}$. Finally,
  \cref{ax:global-points} tells us that $\Modify[\GM]{\Int}$ is equivalent to $\Bool$ and we may
  replace $\prn{\DeclVar{i}{\GM}{\Int}} \to \phi\prn{f^\eta\prn{i}}$ with
  $\prn{b : \Bool} \to \phi\prn{f^\eta\prn{\Con{if}\ b\ \Con{then}\ 0\ \Con{else}\ 1}}$ and
  conclusion follows.
\end{proof}

\begin{lemma}
  Given a category $\DeclVar{C}{\GM}{\Uni}$ and $\DeclVar{\phi}{\GM}{\MFn[\GM]{C}{\Prop_\Simp}}$
  then $\DeclVar{a}{\GM}{C}$ is an element of $C_\phi$ if and only if $\phi\prn{a}$ holds.
\end{lemma}

By choosing different predicates on $\Space$ we obtain a number of familiar
categories. For instance:
\begin{definition}
  The category of $n$-truncated groupoids $\Space_{\le n}$ is given by
  $\Space_{\Con{hasHLevel}\,(n+2)}$.\footnote{The correction $+2$ ensures that $\Space_{\le n}$
    comports with the standard indexing in homotopy theory which begins at $-2$, not $0$.} In
  particular, the category of \emph{propositions} is given by $\Space_{\le -1}$, and the category of
  \emph{sets} is given by $\Space_{\le 0}$.
\end{definition}

\NewDocumentCommand{\FinSet}{}{\mathcal{F}}
\begin{definition}
  The category of finite sets $\FinSet$ is given by $\Space_\phi$ where
  $\phi\prn{X} = \Sum{n : \Nat} \prn{X = \Nat_{\le n}}$.
\end{definition}
\noindent
Note that $\FinSet$ is quite different than $\Sum{A : \Space} \exists n.\,\Nat_{\le n} = A$, which
has only invertible morphisms. The definition of $C_\phi$ is necessary to ensure that $\phi$ is
applied only to the objects of $C$, not its higher cells.

\cref{thm:consequences:fullness} implies that these examples inherit directed univalence from
$\Space$, the first instance of the \emph{structure homomorphism principle
  (SHP)}~\parencite{weaver:2020}: homomorphisms in structured types coincide with their standard
analytic formulations and, consequently, all terms and types are functorial for these analytic
morphisms. For instance, a morphism in $\FinSet$ corresponds to an ordinary function and,
consequently, a family $F : \FinSet \to \Space$ has an action $F\prn{A} \to F\prn{B}$ for any
ordinary function $A \to B$.

\subsection{The structure homomorphism principle}
Not only full subcategories of $\Space$ enjoy SHP, in this section we survey other
categories which satisfy it as well. As a prototypical example, we consider pointed spaces,
$\Space_* = \Sum{A : \Space} A$:
\begin{lemma}
  \label{lem:consequences:pointed-spaces}
  Homomorphisms $\Hom[\Space_*]{\prn{A,a}}{\prn{B,b}}$ are equal to pointed functions
  $\Sum{f : A \to B} f\prn{a} = b$.
\end{lemma}
\begin{proof}
  By \cref{lem:space:space-characterization}, the projection map $\Space_* \to \Space$ is
  covariant, giving, for any pair of pointed spaces $(A,a_0)$ and $(B,b_0)$, an equivalence between
  homomorphisms from $a_0$ to $b_0$ lying over a homomorphism $f : A \to B$ and identifications
  $f(a_0) = b_0$.
\end{proof}

This same methodology can be applied to more general algebraic structures to yield categories of
\eg{}, monoids, groups, rings, \etc{} which we conjecture all enjoy SHP. Rather than dealing with this
generality, we will focus on monoids to complete the example given in \cref{sec:introduction}. We
recall the type of monoids:
\[
  \Con{Monoid} =
  \Sum{A : \Space_{\le 0}}
  \Sum{\epsilon : A}
  \Sum{\cdot : A \times A \to A}
  \Con{isAssociative}\prn{\cdot}
  \times \Con{isUnit}\prn{\cdot,\epsilon}
\]
By repeated application of the closure of categories under dependent sums, functions, and
equalities, we already conclude that $\Con{Monoid}$ is a category. Moreover, we can
characterize its homomorphisms.
\begin{lemma}
  \label{thm:consequences:monoid-dsip}
  A homomorphism
  $\Hom{\prn{A,\epsilon_A,\cdot_A,\alpha_A,\mu_A}}{\prn{B,\epsilon_B,\cdot_B,\alpha_B,\mu_B}}$ is
  precisely a standard monoid homomorphism \eg{} a function $A \to B$ commuting with multiplication
  and the unit.
\end{lemma}

\begin{proof}
To show that the type $\Con{Monoid}$ is a category, we proceed in several steps. Let us denote by $\Space_{\le 0,*}$ the category of pointed sets $\Sum{A : \Space_{\leq 0}} A$. We obtain define the type $\Con{Mag} \defeq \Sum{A : \Space_{\le 0,*}}{A \times A \to A}$ of \emph{magmas} via the pullback, showing that it is a category:
 \[
\begin{tikzpicture}[diagram]

  \node[pullback] (Magpt) {$\Con{Mag}$};
  \node[right=3cm of Magpt] (Setarr) {$\Space_{\le 0}^\Int$};

  \node[below=of Magpt] (Unit) {$\Unit$};
  \node[right=3cm of Unit] (Setpairs) {$\Space_{\ge 0} \times \Space_{\ge 0}$};

  \path[->] (Setarr) edge node[right] {$\prn{s,t}$} (Setpairs);
  \path[->] (Magpt) edge (Setarr);
  \path[->] (Magpt) edge (Unit);
  \path[->] (Unit) edge node[above] {$\prn{A \times A,A}$} (Setpairs);

\end{tikzpicture}
\]
 We get a natural projection from the type $\Con{Mag}_*$ of \emph{pointed magmas} which is also a covariant fibration as evinced by the following diagram:
 \[
\begin{tikzpicture}[diagram]

  \node[pullback] (Magpt) {$\Con{Mag}_*$};
  \node[pullback, right=3cm of Magpt] (Setpt) {$\Space_{\le 0,*}$};
  \node[right=3cm of Setpt] (Spacept) {$\Space_*$};

  \node[below=of Magpt] (Mag) {$\Con{Mag}$};
  \node[right=3cm of Mag] (Set) {$\Space_{\ge 0}$};
  \node[right=3cm of Set] (Space) {$\Space$};

  \path[->] (Magpt) edge (Mag);
  \path[->] (Setpt) edge (Set);
  \path[->] (Spacept) edge (Space);

  \path[->] (Magpt) edge (Setpt);
  \path[>->] (Setpt) edge (Spacept);

  \path[->] (Mag) edge (Set);
  \path[>->] (Set) edge (Space);

\end{tikzpicture}
\]
The type $\Con{Monoid}$ is then a subtype of $\Con{Mag}_*$, so $\Con{Monoid} \simeq \Con{Mag} \times_{\Con{Prop}} \Unit$, and hence $\Con{Monoid}$ is a category.

We compute the (free) morphisms in $\Con{Monoid}$ as follows:
\begin{align*}
  \Int \to \Con{Monoid} & \simeq \Int \to \Sum{A : \Space_{\leq 0}} \Sum{\varepsilon_A : A} \Sum{\mu_A : A \times A \to A} \Con{isAssociative}\prn{\mu_A}
  \times \Con{isUnit}\prn{\mu_A,\epsilon_A} \\
  \simeq  & \prn{E : \Int \to \Space_{\leq 0}} \times \prn{\epsilon : \Prod{i:\Int} E(i)} \times \prn{\mu : \Prod{i : \Int} E(i) \times E(i) \to E(i)} \\
  \times
  & \prn{\Con{isAssociative}\prn{\mu(i)}} \times \Con{isUnit}\prn{\mu(i),\epsilon(i)} \\
  \stackrel{\text{\Cref{thm:space:dua}}}{\simeq}  & \prn{E : \Int \to \Space_{\leq 0}} \times \prn{\epsilon : \Prod{i:\Int} E(i)} \times \prn{\mu : \Prod{i : \Int} \hom_{\Space_{\leq 0}}\prn{E(i) \times E(i),E(i)}} \\
  \times
  & \prn{\Con{isAssociative}\prn{\mu(i)}} \times \Con{isUnit}\prn{\mu(i),\epsilon(i)} \\
   \simeq  & \prn{E : \Int \to \Space_{\leq 0}} \times \prn{\epsilon : \Prod{i:\Int} E(i)} \times \prn{\mu : \Prod{i,j : \Int \times \Int \to \Space_{\leq 0}} \hom_{\Space_{\leq 0}}\prn{E(i) \times E(i),E(i)}} \\
  \times
  & \Prod{i:\Int} \prn{\Con{isAssociative}\prn{\lambda j.\mu(i,j)}} \times \Con{isUnit}\prn{\lambda j.\mu(i,j),\epsilon(i)} \\
  & \stackrel{\substack{\text{\Cref{thm:space:dua}} \\ \text{\Cref{lem:consequences:pointed-spaces}}}}{\simeq} \Sum{\substack{A,B : \Space_{\leq 0} \\ f : A \to B}} \Sum{\substack{\epsilon_A : A \\ \epsilon_B : B}} \prn{f(a) = b} \times \Sum{\substack{\mu_A : A \times A \to A \\ \mu_B : B \times B \to B}} f \circ \mu_A = \mu_B \circ (f \times f) \\ 
  & \times \Con{isAssociative}\prn{\mu_A} \times \Con{isAssociative} \prn{\mu_B} \times \Con{isUnit}\prn{\varepsilon_A} \times \Con{isUnit}\prn{\varepsilon_B}
\end{align*}
In the final step, we have additionally used the characterization of squares $\Int \times \Int \to A$ in a category $A$ as (homotopy) commutative squares of morphisms. For illustration, we can represent the square $\mu$ in terms of the family $E$ as follows:
\begin{center}
  \begin{tikzpicture}[diagram]

  \node (A) {$A$};
  \node[right=3cm of A] (B) {$B$};

  \node[below=2cm of A] (AA) {$A \times A$};
  \node[right=3cm of AA] (BB) {$B \times B$};

  \path[->] (A) edge node[above] {$E$} (B);

  \path[->] (AA) edge node[below] {$E \times E$} (BB);

  \path[->] (AA) edge node[left] {$\mu(0,-)$} (A);
  \path[->] (BB) edge node[right] {$\mu(1,-)$} (B);

\end{tikzpicture}
\end{center}
Overall, the claim follows by taking fibers of $\Con{Monoid}^\Int \to \Con{Monoid} \times \Con{Monoid}$, using the above characterization.

\end{proof}

Substituting \cref{thm:consequences:monoid-dsip} within \cref{thm:consequences:natural}, we obtain
the promised result:
\monoidnat

To complete our goal of proving $\Con{sum}$ natural automatically, it remains only to define
$\Con{List}$ as an endomap of monoids where $\Con{List}\,A$ has pointwise
multiplication. Remarkably, this is straightforward consequence of our results. One need only write
down the definition of this monoid in the ordinary way and conclude that it lifts to a functor
because the carrier ($\Con{List} = \Sum{n : \Nat} -^n$) is already known to be a functor
$\Space_{\le 0} \to \Space_{\le 0}$ using the closure under $\Sigma$ and $\Nat$; no
special argument is required.

We can also apply directed univalence to non-algebraic structures using our ability to define
$n$-presheaf categories $\PSH[n]{C} = \Modify[\OM]{C} \to \Space_{\le n}$. We consider the
representative example of partial orders, which we isolate as a full subcategory of a presheaf
category. In particular, we begin with the category of reflexive graphs:
$\RGraph = \PSH[0]{\Delta_{\le 1}}$ where $\Delta_{\le 1}$ is the ``walking fork'' given by the
pushout $\Delta^2 \Pushout{\Int} \Delta^2$ adjoining a pair of retractions $\partial_0,\partial_1$
to a single arrow $r : 1 \to 0$. While we have not ensured $\Delta_{\le 1}$ is a category, this does
not matter as $\RGraph$ is a category regardless.

We use directed univalence to characterize this category's \emph{objects} as well as its higher
structure:
\begin{lemma}
  The category $\RGraph$ is equivalent to
  $\Sum{G_0 : \Space_{\leq 0}} \Sum{G_1 : G_0 \times G_0 \to \Space_{\leq 0}} \Prod{x:G_0} G_1(x,x).$
\end{lemma}
\begin{proof}
  Using the universal property of a pushout,
  $\RGraph = \Space_{\le 0}^{\Delta^2} \times_{\Space^\Int} \Space_{\le 0}^{\Delta^2}$ and so
  repeated application of \cref{thm:space:dua,thm:space:segal} proves
  $\RGraph = \Sum{G_0\, G_1 : \Space_{\leq 0}} \Sum{s\,t : G_1 \to G_0}\Sum{r : G_0 \to G_1} sr = st$
  and the conclusion now follows from a standard argument.
\end{proof}

We isolate $\Pos \subseteq \RGraph$ as a full subcategory spanned by objects where $G_1$ is a
partial order:

\begin{definition}
  $\Pos = \RGraph_\phi$ where
  $\phi\prn{G} \defeq
  \Con{isASym}\prn{G_1}
  \times
  \Con{isTrans}\prn{G_1}
  \times
  \Prod{x,y : G_0} \IsProp\prn{G_1\prn{x,y}}$
\end{definition}
\cref{thm:consequences:fullness} now proves that homomorphisms in $\Pos$ are precisely monotone maps:
\begin{lemma}
  If $P,Q : \Pos$ then
  $\Hom[\Pos]{P}{Q} \Equiv \Sum{f : P_0 \to Q_0}{\prod_{x,y:P_0}P_1(x,y) \to Q_1(fx,fy)}$.
\end{lemma}
Finally, for the next subsection we isolate a category which is foundational to $\infty$-category
theory: the simplex category $\SIMP$ is the full subcategory $\Pos_\phi$ where
$\phi\prn{P} = \Sum{n : \Nat} P = \Delta^n$.

\subsection{First steps in synthetic higher algebra}

As homotopy (type) theorists like to quip: homotopy types are modern sets. Higher algebra seeks to
take this slogan a step further by studying groups, rings, modules, \etc{} in a world where homotopy
types have replaced sets. While higher algebra has numerous applications to algebraic topology,
algebraic K-theory, and algebraic geometry, it is also a notoriously technical: even the simplest
higher algebraic structure must account for an infinite tower of coherences for each imposed
equation. For our final application of $\Space$, we initiate the study of \emph{higher
  algebra}~\parencite{lurie:2017,gepner:2020} in \TTT{} by defining some of the central objects of
study.  We begin by defining the category of (homotopy-coherent and untruncated) monoids following
\textcite{segal:1974}.

\begin{definition}
  The \emph{category of coherent monoids} $\Con{Monoid}_\infty$ as the full subcategory of $\PSH{\Delta}$
  carved out by the following predicate (the Segal condition):
  \begin{align*}
    &\phi\prn{\DeclVar{X}{\GM}{\Modify[\OM]{\Delta}} \to \Space} =
    \\
    &\quad
    \IsContr\prn{(X(\Delta^0))}
    \times 
    \Prod{n : \Nat_{\ge 1}}
    \IsEquiv\prn{\gl{X\prn{\iota_k}_{k < n}} : X\prn{\Delta^{n}} \to X\prn{\Delta^{1}}^n}
  \end{align*}
  In the above, $\iota_k : \Delta^{1} \to \Delta^{n}$ is $\lambda i.\,\prn{1, \dots, 1, i, 0, \dots}$
  picking out $k$ copies of $1$.
\end{definition}

In other words, a coherent monoid is a functor $X : \Modify[\OM]{\Delta} \to \Space$ such that
$X\prn{\Delta^{n}}$ is the $n$-fold product of $X\prn{\Delta^{1}}$. While somewhat indirect, these
conditions encode all the necessary structure \eg{}, the \emph{unit} is given by the center of contraction $\varepsilon_X : X(\Delta^0)$, and \emph{multiplication} is
given by the composite map $\mu_X : X(\Delta^1)^2 \Equiv X(\Delta^2) \to X(\Delta^1)$.

As a small example of manipulating this definition, we prove the following:
\begin{lemma}
  The functor $\Con{Monoid}_\infty \to \Space$ induced by evaluation at $\Delta^{1}$ reflects
  isomorphisms.
\end{lemma}
\begin{proof}
  Given $f : X \to Y$, by \textcite{riehl:2017} and \cref{thm:consequences:fullness}, it
  suffices to show that if $f\prn{\Delta^{1}}$ is an isomorphism so is $f\prn{\Delta^{n}}$ for any $n$. By
  the Segal condition and naturality, $f\prn{\Delta^{n}}$ is equivalent to
  $\prn{f\prn{\Delta^{1}}}_{i \le n}$ which is invertible if $f\prn{\Delta^{1}}$ is an
  isomorphism.
\end{proof}

Once again, directed univalence yields that the morphisms in $\Con{Monoid}_\infty$ preserve the relevant structure.

\begin{proposition}
  \label{prop:consequences:cohmon}
  Let $\DeclVar{X,Y}{\GM}{\Con{Monoid}_\infty}$. Then $\Hom[\Con{Monoid}_\infty]{X}{Y}$ is equivalent to the type of natural transformations $X \to Y$. In particular, for $\DeclVar{F}{\GM}{\Hom[\Con{Monoid}_\infty]{X}{Y}}$ we have homotopies $F(\Delta^0)(\epsilon_X) = \epsilon_Y$ in $Y(\Delta^0)$, and in the category of spaces $\Space$ the following homotopy-commutative diagram:
  \[
    \DiagramSquare{
      nw = X(\Delta^1)^2 \simeq X(\Delta^2),
      sw = X(\Delta^1),
      ne =  Y(\Delta^2) \simeq Y(\Delta^1)^2,
      se = Y(\Delta^1),
      north = F(\Delta^2),
      south = F(\Delta^1),
      west = \mu_X,
      east = \mu_Y,
      width = 4.5cm
    }
  \]
\end{proposition}

\begin{proof}
  Analogously to \cite[Proposition 6.6]{riehl:2017}, $F$ is a natural transformation. Then preservation of the unit follows by contractibility, and preservation of multiplication follows by directed univalence and naturality.
\end{proof}

We can also define the category of coherent \emph{groups}:
\begin{definition}
  The \emph{category of coherent groups} $\Con{Grp}_\infty$ is the full subcategory of
  $\Con{Monoid}_\infty$ carved out by the predicate
  $\phi{\prn{\DeclVar{X}{\GM}{\Con{Monoid}_\infty}}} = \IsEquiv\prn{\lambda x\,y.\prn{x,\mu\prn{x,y}} : X\prn{\Delta^{1}}^2 \to X\prn{\Delta^{1}}^2}$.
\end{definition}

These concepts and many others can be unified through the formalism of ($\infty$-)operads but we
leave it to future work to develop this apparatus in \TTT{}. An application of such a
formalism would be the ability to develop higher algebra not just in $\Space$, but in
\emph{spectra}, another fundamental category in modern homotopy theory. We conclude this section by
constructing this category.

Suppose $C$ is a pointed category with pullbacks, \ie{}, $C$ has pullbacks and comes with an element
$0 : C$ which is simultaneously initial terminal and initial. Within $C$, we define the \emph{loop
  functor} $\Omega : C \to C$ by $\Omega \defeq \lambda x. 0 \times_x 0$. We have already
encountered such a pointed category: $\Space_*$.
\begin{definition}
  The \emph{category of spectra} $\Sp$ is defined as
  $\Lim_{n : \Nat}
  \prn{\dots \stackrel{\Omega}{\to} \Space_* \stackrel{\Omega}{\to} \Space_* \stackrel{\Omega}{\to} \Space_*}$.
\end{definition}
Here $\lim$ refers to the ordinary definition of a limit from \HOTT{} and we note that as the limit
of categories, $\Sp$ is itself automatically a category. Using directed univalence, we can easily
show that objects of $\Sp$ are infinite deloopings of a groupoid as
expected~\parencite{shulman:2013,vandoorn:phd}.

\section{Conclusions and related work}
\label{sec:conclusions}

We have introduced \TTT{}, an enhancement of simplicial type theory featuring modalities and a
relaxed interval type. We have used \TTT{} as a framework to construct a directed univalent universe
of groupoids $\Space$ which we have further proven to be a well-behaved category. Finally, we have
used $\Space$ as a jumping off point to construct numerous examples of categories and categorical
reasoning in \TTT{} relevant both to $\infty$-category theory and mechanized verification. In order
to do so, we have shown how our same modal operators can be used to \eg{}, construct full
subcategories.

\subsection{Related work}
While directed type theory generally and simplicial type theory specifically are relatively new
areas, there is already substantial work exploring the impact of a ``type theory where types are
categories.'' Much of this work focuses on either constructing such type theories%
~\parencite{%
  licata:2011,%
  warren:2013,%
  nuyts:2015,%
  riehl:2017,%
  north:2018,%
  kavvos:directed:2019,%
  nuyts:2020,%
  ahrens:2023,%
  neumann:2024,%
  neumann:2025%
}
or studying ``formal'' category theory within them%
~\parencite{%
  riehl:2017,%
  weinberger:twosided:2024,%
  weinberger:sums:2024,%
  buchholtz:2023,%
  bardomiano:2024,
  bardomiano:2025%
}
\ie{}, statements which do not use particular closed non-trivial categories but instead quantify
over arbitrary categories. This is distinct from our focus, which has been to combine essentially
off-the-shelf type theories~\parencite{riehl:2017,gratzer:mtt-journal:2021} and to use this combination
to prove facts about the concrete type $\Space$ and types derived thereof. Closely related to this
is the work by Cavallo, Riehl, and Sattler~\parencite{riehl:2018} and \textcite{weaver:2020}, who both
study directed univalence, in respectively simplicial and \emph{bicubical type theory} (\BCTT{}).

\paragraph{Alternative constructions of $\Space$}
Cavallo, Riehl, and Sattler give an alternative construction of $\Space$ in the intended model of
\STT{}, similar to the classical proof due to \textcite{cisinski:2019}. They have argued externally
that this subuniverse satisfies directed univalence and a version of
\cref{lem:space:space-characterization}. However, their work is strictly external and does not
consider how one might integrate $\Space$ within \STT{}. Given that both our universe and theirs
satisfy \cref{lem:space:space-characterization}, they are weakly equivalent and so our results
further show that their universe is \eg{}, a finitely (co)complete category and closed under various
connectives.

\paragraph{Bicubical type theory}
Most closely related to our work is the paper of \textcite{weaver:2020}. Here,
\citeauthor{weaver:2020} consider a variant of \STT{} based on two layers of cubical type theory and
construct a directed-univalent universe in this setting. Their system, \BCTT{}, uses two distinct
interval types: one to account for homotopy type theory and a further layer for the directed
interval. Bicubical type theory is therefore to \TTT{} as cubical type theory is to \HOTT. Moreover,
the approach used by \opcit{} to construct their universe directly inspired our own approach. In
particular, the definition of amazing covariance and our directed glue type are derived from closely
related constructions in \BCTT{}. Moreover, our \emph{directed homomorphism principle} is
elaborating on an idea proposed by \textcite{weaver:2020}.

The two systems, \TTT{} and \BCTT{}, differ in a number of ways. Most importantly, it is conjectured
that \BCTT{} can be formally presented\footnote{\textcite{weaver:2020} do not give a definition of
  \BCTT{} but instead describe the intended model for any such situation. Their model is, however,
constructive and so it is conjectured that such a definition would satisfy canonicity.} so as to
enjoy canonicity and normalization. On the other hand, \TTT{} certainly does not satisfy canonicity.
Thus, \BCTT{} is likely better suited for ``programming'' with directed univalence. However,
\BCTT{}'s categories and groupoids are not expected to be adequate for ordinary $\infty$-categories
or $\infty$-groupoids and so it is not obvious that it can be used for developing synthetic
$\infty$-category theory.

More fundamentally, while they also work within an internal language and we draw on their overall
strategy in \cref{sec:space}, theirs is the internal \emph{extensional} type theory of
$\PSH{\CUBE_{\mathsf{undirected}} \times \CUBE_{\mathsf{directed}}}$ and so they must not only
construct $\Space$ but also the model of base \HOTT{} around it. This substantially complicates some
of their constructions; their versions of \eg{}, covariance, $\Glue$ and so on include details that
are automatically handled when working pervasively with \HOTT.  This model falsifies
\cref{ax:cubes-separate} and so they must introduce an additional set of axioms (the cobar modality)
work around this. We believe both approaches to directed type theory warrant further consideration
to (1) study our results on top of base cubical type theory rather than \HOTT{} and (2) to translate
our new results to their setting. In particular, \opcit{} proves only that $\Space$ is directed
univalent and does not prove \eg{} \cref{thm:space:segal} but we believe our proof, along with those
results in \cref{sec:applications}, can be translated.

\paragraph{Other closely related type theories}
While not about directed type theory, \textcite{myers:2023} also consider a \HOTT{} for simplicial
spaces. We drew inspiration for some of our axioms (\eg{} \cref{ax:cubes-separate}) from them and
expect their other principles will prove useful to \STT{}. Furthermore, \textcite{cherubini:2023}
formulated a version of \cref{ax:sqc} to study synthetic algebraic geometry which led us to its
inclusion in \TTT{}. Finally, \textcite{riley:2024} presents a type theory with a single amazing right
adjoint whose syntax is well-adapted for this situation. We hope that \opcit{} can be generalized
for \TTT{} to yield more usable syntax.

\subsection{Future work}

We isolate three key directions for future work. First, we wish to extend the experimental proof
assistant \RZK{}~\parencite{kudasov:23} with the minimum level of modal reasoning (\eg{}, at least
$\Modify[\GM]{-}$, $\Modify[\SM]{-}$ and $\Modify[\OM]{-}$) to properly axiomatize and work with
$\Space$ as constructed in this paper. We hope to then use this to mechanize
\cref{sec:applications}. Related to this, we hope to give a constructive model of \TTT{} to give a
computational justification of our axioms. We expect this to contribute to a version of \TTT{} with
canonicity and normalization~\parencite{aagaard:2022,gratzer:normalization:2022}.

Second, in forthcoming work~\parencite{gratzer:2025a} we have generalized our construction of $\Space$ to construct the category of (small)
categories $\Cat$ and proven that it is suitably directed univalent~\parencite{cisinski:2022}. While
modalities were required to construct $\Space$, they are required to \emph{state} the properties
of $\Cat$; directed univalence becomes
$\Modify[\GM]{\Int \to \Cat} \Equiv \Modify[\GM]{\Sum{A : \Cat} \Sum{B : \Cat} A \to B}$ because
homomorphisms from $A$ to $B$ must be the \emph{groupoid} of the category of functors $A \to B$, not
the category. Aside from this, our results scale to this more general setting.

Finally, while we discussed presheaf categories in \cref{sec:applications}, we avoided describing
the Yoneda embedding $C \times \Modify[\OM]{C} \to \Space$. While it is possible to construct this
operation, it requires one additional modality (the twisted arrow construction) and, for reasons of
space, we have regretfully chosen to omit it in the present work. In subsequent
work~\parencite{gratzer:2025}\footnote{The chronology of this work is somewhat confused. The
  construction of $\Space$ was completed prior to the cited work on the Yoneda embedding. However,
  publication timing has meant that this second paper was published first.} we have detailed this
additional modality along with the resulting definition of the Yoneda embedding. Using this in
conjunction with our work on full subcategories, we are able to prove various important results
\eg{} that $\Space$ is cocomplete.

In additional forthcoming work~\parencite{gratzer:2025b} we are discussing cocompleteness and show the category of spectra to be stable.

\printbibliography

\end{document}